\documentclass[journal]{IEEEtran}
\ifCLASSINFOpdf

\else

\fi
\usepackage{amsmath}
\usepackage{makeidx}  
\usepackage{algorithm}
\usepackage{algorithmic}
\usepackage{graphicx}
\usepackage{subfigure}
\usepackage{epstopdf}
\usepackage{bm}
\usepackage{cite}
\usepackage{stfloats}

\usepackage{amssymb}
\usepackage{graphicx}

\usepackage{url}

\usepackage{tabularx,booktabs}
\newcolumntype{C}{>{\centering\arraybackslash}X} 
\usepackage{lipsum}

\usepackage{makecell} 

\usepackage{graphicx}
\usepackage{color}
\newtheorem{thm}{Theorem}

\newtheorem{lem}{Lemma}
\newtheorem{pos}{Proposition}
\newtheorem{proof}{proof}

\begin{document}

\title{IRS Aided Federated Learning: Multiple Access and Fundamental Tradeoff}

\author{Guangji Chen,
        Jun Li,
        Qingqing Wu,
        Yiyang Ni,
        and Meng Hua \vspace{-26pt}
        \thanks{Guangji Chen is with Nanjing University of Science and Technology, Nanjing 210094, China (email: guangjichen@njust.edu.cn). Jun Li is with Southeast University, China (jleesr80@gmail.com). Qingqing Wu is with Shanghai Jiao Tong University, 200240, China (e-mail: qingqingwu@sjtu.edu.cn). Yiyang Ni is with the Jiangsu Second Normal University and Jiangsu Institute of Educational Science Research, Nanjing 210094, China (email: niyy@njupt.edu.cn).  M. Hua is with the Department of Electrical and Electronic Engineering, Imperial College London, London SW7 2AZ, UK  (e-mail: m.hua@imperial.ac.uk).}}

\maketitle
\vspace{-3pt}
\begin{abstract}
This paper investigates an intelligent reflecting surface (IRS) aided wireless federated learning (FL) system, where an access point (AP) coordinates multiple edge devices to train a machine leaning model without sharing their own raw data. During the training process, we exploit the joint channel reconfiguration via IRS and resource allocation design to reduce the latency of a FL task. Particularly, we propose three transmission protocols for assisting the local model uploading from multiple devices to an AP, namely IRS aided time division multiple access (I-TDMA), IRS aided frequency division multiple access (I-FDMA), and IRS aided non-orthogonal multiple access (I-NOMA), to investigate the impact of IRS on the multiple access for FL. Under the three protocols, we minimize the per-round latency subject to a given training loss by jointly optimizing the device scheduling, IRS phase-shifts, and communication-computation resource allocation. For the associated problem under I-TDMA, an efficient algorithm is proposed to solve it optimally by exploiting its intrinsic structure, whereas the high-quality solutions of the problems under I-FDMA and I-NOMA are obtained by invoking a successive convex approximation (SCA) based approach. Then, we further develop a theoretical framework for the performance comparison of the proposed three transmission protocols. Sufficient conditions for ensuring that I-TDMA outperforms I-NOMA and those of its opposite are unveiled, which is fundamentally different from that NOMA always outperforms TDMA in the system without IRS. Simulation results validate our theoretical findings and also demonstrate the usefulness of IRS for enhancing the fundamental tradeoff between the learning latency and learning accuracy.
\end{abstract}


\IEEEpeerreviewmaketitle

\vspace{-10pt}
\section{Introduction}
\vspace{-2pt}
Future wireless networks are expected to be key enables for integrating the functions of communication and computation, with the capability to support a vast of internet-of-things (IoT) applications, such as smart grids, augmented and virtual reality, intelligent industry, and autonomous vehicles. This thus calls for deploying artificial intelligence (AI) in wireless networks as a core of operations for various applications \cite{bouzinis2021wireless}. Conventional AI techniques are generally implemented in a centralized manner, which requires the collection of massive amount of data at the central cloud from massive IoT devices. However, collecting big data definitely causes extremely high traffic loads in wireless networks. Benefited by the recent advances of mobile edge computing (MEC), the growing computational capabilities of edge devices paved the way to conduct distributed learning frameworks for the learning models \cite{letaief2021edge}. Among the distributed learning approaches, federated learning (FL) has shown its potential in satisfying the low-latency demands and preserving the data privacy \cite{li2020federated,10177379}. In a wireless FL system, devices collaboratively train shared learning models by exploiting their local data samples with the coordination of an access point (AP).

The implementation of wireless FL requires the model exchange between the AP and massive devices over hundreds of rounds to achieve a satisfied learning accuracy. Hence, wireless FL faces several critical challenges, which arise from limited communication-computation resources, e.g., the scarce spectrum resources, low energy budget of devices, and unreliable wireless channels \cite{imteaj2021survey}. In light of these issues, substantial works dedicated on the joint communication-computation resource allocation design to improve the performance of wireless FL in terms of various performance metrics, such as learning accuracy \cite{chen2020joint, liu2021joint, ren2020scheduling,wei2024gradient}, convergence latency \cite{chen2020convergence, bouzinis2023wireless, fu2023federated}, and energy efficiency \cite{yang2020energy, kim2022tradeoff, mo2021energy}. In particular, the work \cite{chen2020convergence} developed a joint communication-learning framework to minimize the training loss by jointly optimizing the device scheduling and transmit power based the derived convergence bound. For another challenge of the communication burden induced by the model uploading of massive devices, the exploitation of advanced multiple access schemes can facilitate the improvement of the efficiency for wireless FL \cite{bouzinis2022wireless}. Regarding the multiple access for model uploading, the works \cite{mo2021energy, bouzinis2022wireless} demonstrated the superiority of non-orthogonal multiple access (NOMA) over time division multiple access (TDMA) in terms of both the objectives of latency and energy consumption.

Despite the above theoretical progress, relying only on the communication-computation resource optimization \cite{chen2020joint, liu2021joint, ren2020scheduling,wei2024gradient,chen2020convergence, bouzinis2023wireless, fu2023federated, yang2020energy} may not guarantee the performance of wireless FL due to the severe wireless fading. A large number of communication stragglers arise from unfavorable channels in a FL system may not be scheduled to participate in the training process due to the strict latency requirement, which fundamentally limits the full potential of FL in terms of data utilization \cite{imteaj2021survey}. Although exploiting the high beamforming gain attained by massive multiple-input multiple-output (MIMO) \cite{mu2022federated} can effectively alleviate the straggler issue, it still faces practical challenges, e.g., exceedingly high hardware cost and energy consumption. As a remedy, intelligent reflecting surface (IRS) has been emerged as a cost-effective technology to customize favorable channels via smartly passive refection \cite{wu2021intelligent, di2020smart, chen2023fundamental, liu2021reconfigurable}. Particularly, IRSs are digitally-controlled meta-surfaces comprising massive reflecting elements, which can be tuned dynamically to alter phase of the incident signals and thereby creating a ``smart radio environment''. Besides, other practical advantages of IRSs, such as light weight, low profile, and conformal geometry, make them convenient to be deployed in future wireless networks \cite{di2020smart}. To fully reap the potential benefits provided by the IRS, it is of paramount significance to appropriately optimize the IRS phase-shifts so that the favorable wireless propagation environment is reconfigured for enhancing signal transmission. This new research paradigm has been widely investigated in various wireless applications, e.g.,  multi-user MIMO networks \cite{zhi2021statistical, chen2024intelligent}, MEC systems \cite{bai2021resource, chen2022irs, 10316588,chen2022irs1}, NOMA \cite{mu2021capacity, chen2022active}, integrated sensing and communication \cite{meng2022intelligent,10643002}, and over the air computation (AirComp) \cite{fang2021over, zhai2023simultaneously, chen2024intelligent1}.

In addition to the above applications, it is also appealing to make use of the IRS's high passive beamforming gain for improving the performance of wireless FL. Specifically, by placing IRSs in the vicinity of edge devices, the stragglers' communication qualities can be improved significantly due to the intelligent reflections of IRSs. The mitigation of the straggler issue is beneficial for enhancing the training data exploitation, which is of paramount importance to unlock its full potential for achieving the high-quality learning performance in the wireless edge network. To reap the aforementioned benefits, there are two research lines on IRS aided wireless FL, namely, IRS aided AirComp based FL \cite{wang2021federated, ni2022integrating,zhao2023performance,liu2021reconfigurableFL} and IRS aided digital FL \cite{sun2024reconfigurable, zhang2024convergence, li2023reconfigurable}. Specifically, the first research line \cite{wang2021federated, ni2022integrating,zhao2023performance,liu2021reconfigurableFL} aims to exploit the high passive beamforming gain of IRS to suppress the model aggregation error at the AP, where the waveform superposition of the multiple access channels is employed to achieve fast model aggregation in the analog transmission mode. In contrast, IRS aided digital FL \cite{sun2024reconfigurable, zhang2024convergence, li2023reconfigurable} focuses on the joint optimization of IRS phase-shifts and resource allocation to improve the communication-efficiency of the model uploading links in the digital communication mode.

Despite of these works, several fundamental issues still remain unsolved for the IRS aided digital FL. First, whether employing NOMA based model uploading for wireless FL still outperforms orthogonal multiple access (OMA) or not by considering the channel reconfiguration via the IRS? Benefited by the dynamic phase-shifts adjustment of the IRS, artificial time-varying channels are created by designing IRS phase-shifts over time, which facilitates the utilization of the multi-user diversity. By considering the time-selectivity of the IRS, dedicated IRS reflection pattern can be allocated for each individual device for improving the channel quality under the TDMA-based uploading, whereas a shared IRS reflection pattern is employed for assisting the simultaneous model uploading of all devices via NOMA. Different from existing works on the wireless FL without IRS \cite{mo2021energy, bouzinis2022wireless}, the conclusion on NOMA versus OMA needs to be revisited in the IRS aided FL system. Second, how many IRS elements are needed to enable a full scheduling for all devices under a given latency requirement? It is generally believed that leveraging the smart reflection of the IRS is able to enhance the tradeoff between the learning accuracy and learning latency. This is an essential consideration for understanding the fundamental limit of using IRS to address the straggler issue in the digital wireless FL system.

Motivated by the above issues, we investigate an IRS aided wireless FL system by considering both the NOMA and OMA, where an IRS is deployed to enhance the model uploading from multiple devices to the AP, as shown in Fig. 1. We propose three types of transmission protocols, namely IRS aided TDMA (I-TDMA), IRS aided FDMA (I-FDMA), and IRS aided NOMA (I-NOMA), to pursue a theoretical performance comparison between NOMA and OMA in the IRS aided FL system. Particularly, dedicated IRS reflection pattern is allocated to each scheduled device for the I-TDMA, whereas all the scheduled devices share the same set of IRS phase-shifts under the I-FDMA and I-NOMA. Our objective is to minimize the per-round learning latency subject to a given training loss by jointly optimizing the device scheduling, IRS phase-shifts, and communication-computation resource allocation. Different from conventional FL system without IRS where wireless channels are uncontrollable and remain static in a channel coherence block, favourable time-varying channels can be proactively generated in an IRS aided FL system to enhance the multi-user diversity, which thus enables a flexible resource allocation and has a significant impact on the performance of multiple access schemes. The main contributions are summarized as follows.

\begin{figure}[!t]
\centering
\includegraphics[width= 0.5\textwidth]{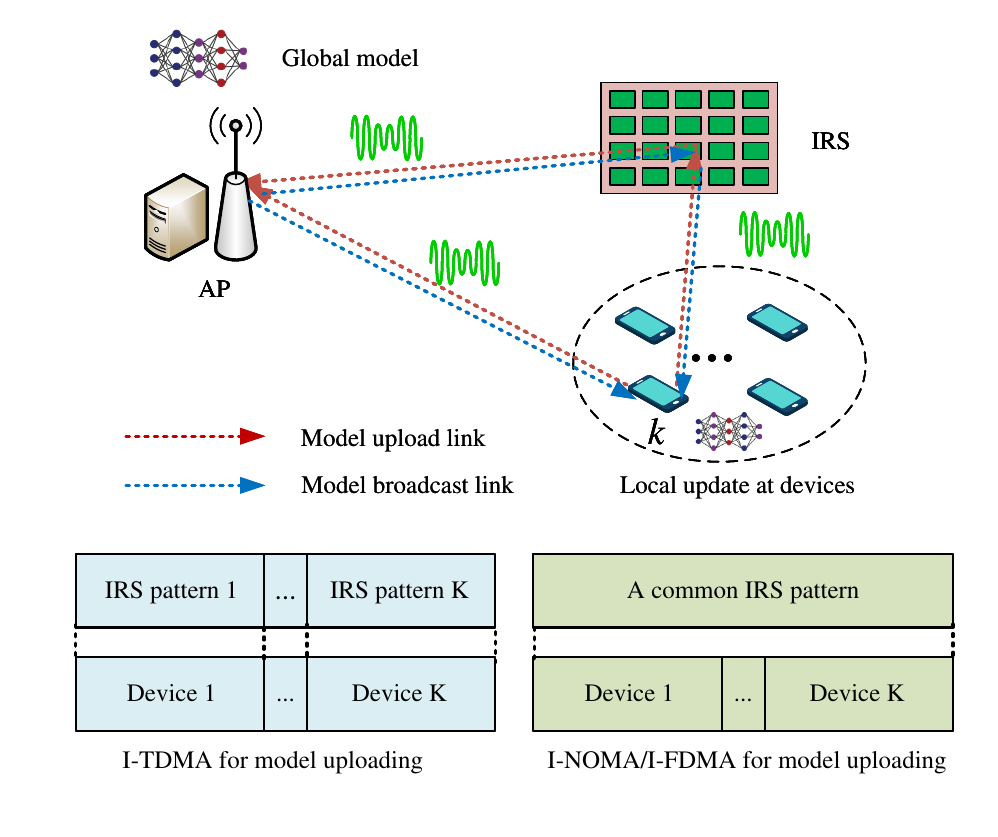}
\DeclareGraphicsExtensions.
\vspace{-28pt}
\caption{IRS aided wireless federated learning with OMA and NOMA.}
\label{model}
\vspace{-12pt}
\end{figure}

\begin{itemize}
  \item First, we propose efficient algorithms to solve the associated latency minimization problems under the three proposed transmission protocols. For the I-TDMA scheme, the device scheduling policy is derived by capturing both the effects of the IRS aided channels and the amount of local data samples. Then, the optimal solution of I-TDMA is obtained based on the device scheduling policy. For the I-FDMA and I-NOMA, we propose efficient algorithms to obtain their high-quality solutions by invoking successive convex approximation (SCA) techniques.
  \item Next, to shed light on the fundamental tradeoff between the learning latency and learning accuracy, we extend the above optimization framework to solve the problems of training loss minimization subject to a given latency requirement. We further derive the required number of IRS elements to enable a full scheduling for all devices, which unveils the usefulness of the IRS to achieve a lossless FL model.
  \item Finally, we provide a theoretical framework for the performance comparison for the proposed three transmission protocols. We unveil that both the I-TDMA and I-NOMA outperform I-FDMA, whereas the comparison results of the I-TDMA and I-NOMA depend on the specific channel structures. Particularly, sufficient conditions for ensuring that I-TDMA outperforms I-NOMA and those of its opposite are derived. Simulation results corroborate our theoretical findings and also demonstrate the superiority of the proposed scheme over the benchmark schemes.
\end{itemize}

The rest of this paper is organized as follows. Section II presents the system model of the IRS aided wireless FL. In Section III, we provide the convergence anaylsis and problem formulations. Section IV proposes efficient algorithms to solve the formulating problems and discusses the usefulness of IRS to enable a full device scheduling. Section V provides a theoretical performance comparison of the proposed three transmission protocols. Section VI provides simulation results to evaluate the proposed designs. Finally, we conclude in Section VII.
\vspace{-10pt}
\section{System model}
As illustrated in Fig. \ref{model}, we consider an IRS aided FL system, which consists of a single antenna AP, an IRS, and $K$ single-antenna edge devices. The IRS is equipped with $N$ elements. For convenience, the sets of IRS elements and edge devices are denoted by ${\cal N} \buildrel \Delta \over = \left\{ {1, \ldots ,N} \right\}$ and ${\cal K} \buildrel \Delta \over = \left\{ {1, \ldots ,K} \right\}$, respectively. In this system, the AP with an edge server coordinates $K$ edge devices to learn a shared ML model by exploiting their local data. Let ${{\cal D}_k}$ denote the local dataset of device $k$ with size ${D_k}$. The global model at the AP is denoted by ${\bf{w}} \in {\mathbb{R}^d}$. Accordingly, the local loss function at device $k$ and the global loss function at the AP can be expressed as
\begin{align}\label{local_loss_function}
{F_k}\left( {\bf{w}} \right) = \frac{1}{{{D_k}}}\sum\nolimits_{i = 1}^{{D_k}} {{f_i}\left( {\bf{w}} \right)}, F\left( {\bf{w}} \right) = \frac{1}{D}\sum\nolimits_{k = 1}^K {{F_k}\left( {\bf{w}} \right)}.
\end{align}
respectively, where ${{f_i}\left( {\bf{w}} \right)}$ denotes the loss function of each data sample $i \in {{\cal D}_k}$ and $D = \sum\nolimits_{k = 1}^K {{D_k}}$ represents the total size of data. For edge FL, the learning objective is to obtain a desirable ML model that minimizes the global loss function, i.e.,
\begin{align}\label{learning_task}
{{\bf{w}}^*} = \mathop {\arg \min }\limits_{\bf{w}} F\left( {\bf{w}} \right).
\end{align}
To this end, we adopt batch gradient descent \cite{chen2020joint} which is implemented in a distributed way for updating local model. In particular, the global model ${\bf{w}}$ is updated iteratively with $T$ training rounds, denoted by ${\cal T} \buildrel \Delta \over = \left\{ {1, \ldots ,T} \right\}$. For the $t$-th round, $t \in {\cal T}$, the detailed procedure is implemented as follows:

1) \emph{Device scheduling}: The AP determines a subset of edge devices ${\cal K}_a^t \subseteq {\cal K}$  to participate in the learning process. Let $a_k^t \in \left\{ {0,1} \right\}$ denote the scheduling variable of device $k$. Therein, $a_k^t = 1$ indicates that device $k$ is scheduled by the AP in the $t$-th round; otherwise we have $a_k^t = 0$. Accordingly, the set of scheduled devices can be expressed as ${\cal K}_a^t \buildrel \Delta \over = \left\{ {k:a_k^t = 1,k \in {\cal K}} \right\}$.

2) \emph{Global model broadcast}: Then, the AP broadcasts the global model obtained in the previous round, denoted by ${{\bf{w}}_{t - 1}}$, to all the scheduled devices.

3) \emph{Local model computation}: After receiving the global model, each scheduled device $k \in {\cal K}_a^t$ computes its local model ${{\bf{w}}_{k,t}}$ by employing the gradient descent algorithm based on its local dataset:
\begin{align}\label{local_model_computation}
{{\bf{w}}_{k,t}} = {{\bf{w}}_{t - 1}} - \eta \nabla {F_k}\left( {{{\bf{w}}_{t - 1}}} \right),\forall k \in {\cal K}_a^t,
\end{align}
where $\eta$ denotes the learning rate and ${\nabla {f_i}\left( {{{\bf{w}}_{t - 1}}} \right)}$ represents the gradient of ${{f_i}\left( {{{\bf{w}}_{t - 1}}} \right)}$ with respect to ${\bf{w}}$ at the point ${{{\bf{w}}_{t - 1}}}$.

4) \emph{Local model uploading and aggregation}: Each scheduled device $k \in {\cal K}_a^t$ uploads its updated local model ${{\bf{w}}_{k,t}}$ to the AP via the wireless channel. After receiving all scheduled devices' local models, the AP aggregates them to obtain the updated global model as
\begin{align}\label{global_model_computation}
{{\bf{w}}_t} = \frac{{\sum\nolimits_{k = 1}^K {a_k^t{D_k}{{\bf{w}}_{k,t}}} }}{{\sum\nolimits_{k = 1}^K {a_k^t{D_k}} }},
\end{align}
where ${{\bf{w}}_t}$ is the global model obtained at the AP in the $t$-th round.

Considering that the AP has a higher computation capability than that of the edge devices, the time used to global model aggregation is negligible compared to that of local  model computation. Additionally, the time for  global model broadcasting is also much less than that for local model uploading since the AP has a higher transmit power while also the same information is broadcasted to all devices. Hence, we focus on the stage of local model computation and local model uploading. In the next subsection, we introduce the associated computation and communication process of FL over IRS aided networks.

\vspace{-10pt}
\subsection{Computation Model of Local Computation}
Let ${C_k}$ denote the number of CPU cycles required for computing one data sample of device $k$, which can be obtained as a prior by measuring it offline. Hence, the required number of CPU cycles for running one local round is ${C_k}{D_k}$. Let ${f_k}$ denote the CPU frequency at device $k$. Then, the computation time of device $k$ in one local round can be expressed as
\begin{align}\label{local_computation_time}
\tau_k^{{\rm{loc}}} = \frac{{{C_k}{D_k}}}{{{f_k}}}, \forall k \in {\cal K}.
\end{align}
Accordingly, the energy consumption of device $k$ for the local model computation is
\begin{align}\label{local_computation_energy}
E_k^{{\rm{loc}}} = \xi f_k^3\tau_k^{{\rm{loc}}} = \xi {C_k}{D_k}f_k^2,\forall k \in {\cal K},
\end{align}
where $\xi$ is a constant related to the hardware architecture of device $k$. Similar to existing works \cite{bouzinis2023wireless, fu2023federated, yang2020energy}, the synchronous operation is considered and thereby all scheduled devices train their local models simultaneously. Therefore, the time consumed for local model computation in the $t$-th training round can be written as
\begin{align}\label{local_time}
\tau_t^{{\rm{loc}}} = \mathop {\max }\limits_{k \in {\cal K}} \left\{ {a_k^t\tau_k^{{\rm{loc}}}} \right\}.
\end{align}
\vspace{-10pt}
\subsection{Communication Model of Local Model Uploading}
To characterize the achievable performance upper bound of the IRS aided local model uploading, we assume that the CSI of the involved channels can be obtained perfectly at the AP by employing the channel acquisition schemes discussed in \cite{wu2021intelligent}. The equivalent channels from device $k$ to the AP, from device $k$ to the IRS, and from the IRS to the AP are denoted by ${h_{d,k}} \in \mathbb{C}$, ${{\bf{h}}_{r,k}} \in { \mathbb{C}^{N \times 1}}$, and ${{\bf{g}}^H} \in {\mathbb{C}^{1 \times N}}$, respectively, where $k \in {\cal K}$.

We aim to investigate a theoretical performance comparison between OMA and NOMA in an IRS aided FL system, where two types of OMA schemes, i.e, TDMA and FDMA, are considered. To this end, three IRS aided multiple access schemes are considered for the local model uploading, namely I-TDMA, I-FDMA, and I-NOMA, respectively.
\subsubsection{I-TDMA Based Transmission}
For the I-TDMA based transmission scheme, the scheduled devices transmit their updated local models to the AP over orthogonal time slots (TSs). Let $\tau _k^{{\mathop{\rm c}\nolimits} }$ denote the local model uploading duration for device $k \in {\cal K}$ and thereby the total transmission duration is $\tau _{\rm{T}}^{{\rm{c}}} = \sum\nolimits_{k = 1}^K {\tau _k^{{\rm{c}}}}$. For the $k$-th TS of device $k$, a dedicated IRS reflection pattern, denoted by ${{\bf{\Theta }}_k} = {\mathop{\rm diag}\nolimits} \left( {{e^{j{\theta _{k,1}}}}, \ldots ,{e^{j{\theta _{k,N}}}}} \right)$, is employed to assist its uplink transmission, where ${\theta _{k,n}} \in \left[ {0,2\pi } \right),\forall k,n$. Denoting the transmit power of device $k$ as $p_k$ and then the achievable data rate of device $k$ under I-TDMA is given by
\begin{align}\label{rate_k_TDMA}
r_k^{\rm{T}} = B{\log _2}\left( {1 + \frac{{{p_k}{{\left| {{h_{d,k}} + {{\bf{g}}^H}{{\bf{\Theta }}_k}{{\bf{h}}_{r,k}}} \right|}^2}}}{{B{\sigma ^2}}}} \right), \forall k \in {\cal K},
\end{align}
where ${{\sigma ^2}}$ is the noise power density and $B$ denotes the total bandwidth. To successfully upload the local model of device $k$ in the $k$-th TS, the condition $r_k^{\rm{T}}\tau _k^{{\rm{c}}} \ge sa_k^t,\forall k \in {\cal K}$
should be satisfied, where $s$ denotes the local model size (i.e., the number of bits) of each device. Accordingly, the energy consumption of device $k$ for local model uploading is $E_{k,{\rm{T}}}^{{\rm{com}}} = \tau _k^{{\rm{c}}}{p_k},\forall k \in {\cal K}$.
\subsubsection{I-FDMA Based Transmission}
For the I-FDMA based transmission scheme, the total bandwidth $B$ is partitioned into multiple orthogonal sub-bands and each of them is assigned to a device, with a bandwidth of ${b_k}B$, $\forall k \in {\cal K}$, where ${b_k}$ denotes the bandwidth allocation coefficient and satisfies $\sum\nolimits_{k = 1}^K {{b_k} \le 1}$. Let $\tau _{\rm{F}}^{{\rm{c}}}$ denote the total transmission duration of I-FDMA. During the time interval of $\tau _{\rm{F}}^{{\rm{c}}}$, a common IRS refection pattern ${\bf{\Theta }} = {\mathop{\rm diag}\nolimits} \left( {{e^{j{\theta _1}}}, \ldots ,{e^{j{\theta _N}}}} \right)$ with ${\theta _n} \in \left[ {0,2\pi } \right),\forall n$, is employed. Accordingly, the achievable data rate of device $k$ under I-FDMA can be expressed as
\begin{align}\label{rate_k_FDMA}
r_k^{\rm{F}} \!=\! {b_k}B{\log _2}\left( {1 + \frac{{{p_k}{{\left| {{h_{d,k}} + {{\bf{g}}^H}{\bf{\Theta }}{{\bf{h}}_{r,k}}} \right|}^2}}}{{{b_k}B{\sigma ^2}}}} \right),\forall k \in {\cal K},
\end{align}
Then, the condition $r_k^{\rm{F}}\tau _{\rm{F}}^{{\rm{c}}} \ge sa_k^t$
should be satisfied to ensure that the local model of device $k$ can be uploaded successfully. The associated energy consumption of each device under the I-FDMA is $E_{k,{\rm{F}}}^{{\rm{com}}} = \tau _{\rm{F}}^{{\rm{c}}}{p_k},\forall k \in {\cal K}$.

\subsubsection{I-NOMA Based Transmission}
For the I-NOMA based transmission scheme, all scheduled devices transmit their local models to the AP simultaneously. The total transmission duration of I-NOMA is denoted by $\tau _{\rm{N}}^{{\rm{c}}}$ and a common IRS reflection pattern ${\bf{\Theta }}$ is shared by all scheduled devices. By employing the SIC technique at the AP and allowing time sharing among different decoding orders, the achievable rate region of all devices under I-NOMA is given by
\begin{align}\label{capacity_region}
{{\cal R}_{\rm{N}}}\left( {{\bf{p}},{\bf{\Theta }}} \right) \!\!=\!\! \left\{ \begin{array}{l}
{{\bf{r}}^N} \in \mathbb{R}_ + ^{K \times 1}:\sum\limits_{k \in {\cal J}} {r_k^N} \\
 \le B{\log _2}\left( {1 + \frac{{\sum\nolimits_{k \in {\cal J}} {{p_k}{{\left| {{h_k}\left( {\bf{\Theta }} \right)} \right|}^2}} }}{{B{\sigma ^2}}}} \right),\forall {\cal J} \subseteq {\cal K}
\end{array} \right\},
\end{align}
where ${\bf{p}} \buildrel \Delta \over = {\left[ {{p_1}, \ldots {p_K}} \right]^T}$, ${{\bf{r}}^N} \buildrel \Delta \over = {\left[ {r_1^{\rm{N}}, \ldots r_K^{\rm{N}}} \right]^T}$, ${h_k}\left( {\bf{\Theta }} \right) = {h_{d,k}} + {{\bf{g}}^H}{\bf{\Theta }}{{\bf{h}}_{r,k}}$, and ${\cal J}$ represents any subset contained in set ${\cal K}$. Under the I-NOMA scheme, the corresponding condition for that device $k$ successfully upload its local model is $r_k^{\rm{N}}\tau _{\rm{N}}^{{\rm{c}}} \ge sa_k^t$.
The communication energy of I-NOMA at device $k$ is thus given by $E_{k,{\rm{N}}}^{{\rm{c}}} = \tau _{\rm{N}}^{{\rm{c}}}{p_k},\forall k \in {\cal K}$.
\vspace{-10pt}
\section{Convergence Analysis and Problem Formulation}
This section provides the convergence analysis of the IRS aided FL, which captures the impact of the number of scheduled devices on the convergence performance. Based on the convergence analysis, the per-round latency minimization problems are formulated under proposed three transmission schemes.
\vspace{-10pt}
\subsection{Convergence Analysis}
Following the existing work \cite{liu2021reconfigurableFL}, we make several assumptions regarding the loss function and  local gradients.

\emph{Assumption 1 (Smoothness)}: The loss function $F\left( {\bf{w}} \right)$ is $L$-smooth, i.e., $\forall {{\bf{w}}_1},{{\bf{w}}_2} \in {\mathbb{R}^d}$, there exists a non-negative constant $L$, such that
\begin{align}\label{assumption1}
F\left( {{{\bf{w}}_2}} \right)\!\! \le\!\! F\left( {{{\bf{w}}_1}} \right) \!\!+\!\! \nabla F{\left( {{{\bf{w}}_1}} \right)^T}\left( {{{\bf{w}}_2} \!\!-\!\! {{\bf{w}}_1}} \right) \!\!+\!\! \frac{L}{2}\left\| {{{\bf{w}}_2} \!\!-\!\! {{\bf{w}}_1}} \right\|_2^2.
\end{align}

\emph{Assumption 2 (Polyak-Lojasiewicz Inequality)}: Let ${F^*}$ denote the optimal objective value of problem \eqref{learning_task}. There exists a positive constant $\delta  > 0$ such that $F\left( {\bf{w}} \right), \forall {\bf{w}} \in {\mathbb{R}^d}$, satisfies
\begin{align}\label{assumption2}
\left\| {\nabla F\left( {\bf{w}} \right)} \right\|_2^2 \ge 2\delta \left( {F\left( {\bf{w}} \right) - {F^*}} \right).
\end{align}

\emph{Assumption 3 (Bounded Gradient Norm)}: There exists a positive constant $\varepsilon > 0$ such that the sample-wise gradients of edge devices are bounded by
\begin{align}\label{assumption3}
\left\| {\nabla {f_i}\left( {\bf{w}} \right)} \right\|_2^2 \le \varepsilon ,\forall i \in {{\cal D}_k}, k \in {\cal K}.
\end{align}

Based on the above assumptions, we aim to characterize the optimality gap of a FL task after $T$ training rounds. To this end, we first introduce the following lemma.
\begin{lem}
With $\eta  = 1/L$, we have
\begin{align}\label{gap1}
F\left( {{{\bf{w}}_{t + 1}}} \right) \le F\left( {{{\bf{w}}_t}} \right) - \frac{1}{{2L}}\left\| {\nabla F\left( {{{\bf{w}}_t}} \right)} \right\|_2^2 + \frac{1}{{2L}}\left\| {{{\bf{e}}_t}} \right\|_2^2,
\end{align}
\end{lem}
where
\begin{align}\label{gradient_error}
{{\bf{e}}_t} = \nabla F\left( {{{\bf{w}}_t}} \right) - \frac{{\sum\nolimits_{k = 1}^K {a_k^t{D_k}\nabla {F_k}\left( {{{\bf{w}}_t}} \right)} }}{{\sum\nolimits_{k = 1}^K {a_k^t{D_k}} }}
\end{align}
denotes the gradient error due to the device scheduling.
\begin{proof}
Please refer to Lemma 2.1 in \cite{friedlander2012hybrid}.
\end{proof}
With the help of Lemma 1, an upper bound of the optimality gap is provided in the following theorem.
\begin{thm}
With assumptions 1-3 and $\eta  = 1/L$, the optimality gap after $T$ training rounds is upper bounded by
\begin{align}\label{optimality_gap}
&F\left( {{{\bf{w}}_T}} \right) - {F^*}\nonumber\\
&\le {\left( {1 - \frac{\delta }{L}} \right)^T}\left( {F\left( {{{\bf{w}}_0}} \right) - {F^*}} \right) + \sum\limits_{t = 1}^T {{A_t}{{\left( {1 - \frac{\delta }{L}} \right)}^{T - t}}} ,
\end{align}
where
\begin{align}\label{e_t_bound}
{A_t} = \frac{{2\varepsilon }}{{L{D^2}}}{\left( {\sum\nolimits_{k = 1}^K {\left( {1 - a_k^t} \right){D_k}} } \right)^2}.
\end{align}
\end{thm}
\begin{proof}
The proof is similar to Appendix B of \cite{liu2021reconfigurableFL}.
\end{proof}

It is observed from Theorem 1 that the first term in \eqref{optimality_gap} tends to zero with rate ${\cal{O}}\left( {1/T} \right)$. However, the second term in \eqref{optimality_gap} is related to the model error induced by the device scheduling, which creates a additional gap between $F\left( {{{\bf{w}}_T}} \right)$ and ${F^*}$. By scheduling all edge devices, i.e., $a_k^t = 1,\forall k$, the second term in \eqref{optimality_gap} tends to be zero and the optimality gap in Theorem 1 reduces to that of a lossless FL model. However, scheduling all edge devices may increase the communication latency in the local model uploading stage. Hence, there exists a fundamental tradeoff between the learning accuracy and the total convergence time, which will be studied later.
\vspace{-10pt}
\subsection{Problem Formulation}
Given the number of learning rounds $T$, our goal is to minimize the overall latency of the FL task, under the optimality gap requirements and the energy constraints. Since the future CSI is unknown, we focus on minimizing the latency in each learning round under the availability of the current CSI. Moreover, it is noticed that the upper bound of the optimality gap in Theorem 1 monotonically increases with respect to ${A_t}$. It is obvious that increasing the number of scheduled devices can reduce the value of ${A_t}$, thereby suppressing the associated optimality gap. However, scheduling more devices results in both the increased transmission latency and the computational latency. Taking these issues into consideration, it is crucial to balance the fundamental tradeoff between the learning accuracy and learning latency. To this end, the per-round latency minimization problem by jointly optimizing the IRS phase-shift and the communication-computation resource allocation under the I-TDMA scheme is formulated as follows
\begin{subequations}\label{C1}
\begin{align}
\label{C1-a}\mathop {\min }\limits_{\left\{ {{{\bf{\Theta }}_k}} \right\},{\bf{a}},{{\bm{\tau }}^{\rm{T}}},{\bf{p}},{\bf{f}}} &\sum\nolimits_{k = 1}^K {\tau _k^{{\rm{c}}}}  + \tau _t^{{\rm{loc}}}\\
\label{C1-b}{\rm{s.t.}}\;\;\;\;\;\;&{p_k}\tau _k^{{\rm{c}}} + \xi {C_k}{D_k}f_k^2 \le E_k^{\max },~\forall k,\\
\label{C1-c}&B\tau _k^{{\rm{c}}}{\log _2}\left( {1 \!+\! \frac{{{p_k}{{\left| {{h_k}\left( {{{\bf{\Theta }}_k}} \right)} \right|}^2}}}{{B{\sigma ^2}}}} \right) \!\ge \! sa_k^t, \forall k,\\
\label{C1-d}&\tau _t^{{\rm{loc}}} \ge a_k^t\frac{{{C_k}{D_k}}}{{{f_k}}}, ~\forall k,\\
\label{C1-e}&\frac{{2\varepsilon }}{{L{D^2}}}{\left( {\sum\nolimits_{k = 1}^K {\left( {1 - a_k^t} \right){D_k}} } \right)^2} \le \kappa,\\
\label{C1-f}&{p_k} \ge 0, {f_k} \ge 0, \forall k,\\
\label{C1-g}&a_k^t \in \left\{ {0,1} \right\},~\forall k,\\
\label{C1-h}&\left| {{{\left[ {{{\bf{\Theta }}_k}} \right]}_{n,n}}} \right| = 1,~\forall k,\forall n,
\end{align}
\end{subequations}
where ${h_k}\left( {{{\bf{\Theta }}_k}} \right) = {h_{d,k}} + {{\bf{g}}^H}{{\bf{\Theta }}_k}{{\bf{h}}_{r,k}}$, ${\bf{a}} = \left[ {a_1^t, \ldots ,a_K^t} \right]$, ${{\bm{\tau }}^{\rm{T}}} = \left[ {\tau _1^{\rm{c}}, \ldots ,\tau _K^{\rm{c}},\tau _t^{{\rm{loc}}}} \right]$, ${\bf{f}} = \left[ {{f_1}, \ldots ,{f_K}} \right]$, and ${\bf{p}} = \left[ {{p_1}, \ldots ,{p_K}} \right]$. For problem \eqref{C1}, constraint \eqref{C1-b} indicates devices' energy budget. Constraint \eqref{C1-c} is obtained according to \eqref{rate_k_TDMA} and the condition $r_k^{\rm{T}}\tau _k^{{\rm{c}}} \ge sa_k^t,\forall k \in {\cal K}$. Constraint \eqref{C1-d} specifies the local training latency requirement, and constraint \eqref{C1-e} characterizes the target learning accuracy requirement. The corresponding per-round latency minimization problems with I-FDMA and I-NOMA can be similarly formulated as
\begin{subequations}\label{C2}
\begin{align}
\label{C2-a}\mathop {\min }\limits_{{\bf{\Theta }},{\bf{a}},{\bf{b}},{{\bm{\tau }}^{\rm{F}}},{\bf{p}},{\bf{f}}} \;&\tau _{\rm{F}}^{\rm{c}} + \tau _t^{{\rm{loc}}}\\
\label{C2-b}{\rm{s.t.}}\;\;\;\;\;\;&{p_k}\tau _{\rm{F}}^{\rm{c}} + \xi {C_k}{D_k}f_k^2 \le E_k^{\max },~\forall k,\\
\label{C2-c}&B{b_k}\tau _{\rm{F}}^{\rm{c}}{\log _2}\left( {1 \!+\! \frac{{{p_k}{{\left| {{h_k}\left( {{{\bf{\Theta }}}} \right)} \right|}^2}}}{{B{b_k}{\sigma ^2}}}} \right) \!\ge\! sa_k^t, \forall k,\\
\label{C2-d}&\sum\nolimits_{k = 1}^K {{b_k} \le 1},{b_k} \ge 0, ~\forall k,\\
\label{C2-e}&\left| {{{\left[ {{{\bf{\Theta }}}} \right]}_{n,n}}} \right| = 1,~\forall n,\\
\label{C2-f}&\eqref{C1-d},\eqref{C1-e},\eqref{C1-f},\eqref{C1-g},
\end{align}
\end{subequations}
\vspace{-8pt}
\begin{subequations}\label{C3}
\begin{align}
\label{C3-a}\mathop {\min }\limits_{{\bf{\Theta }},{\bf{a}},{{\bm{\tau }}^{\rm{N}}},{\bf{p}},{\bf{f}}} \;&\tau _{\rm{N}}^{\rm{c}} + \tau _t^{{\rm{loc}}}\\
\label{C3-b}{\rm{s.t.}}\;\;\;\;\;\;&{p_k}\tau _{\rm{N}}^{\rm{c}} + \xi {C_k}{D_k}f_k^2 \le E_k^{\max },~\forall k,\\
\label{C3-c}&{{\bf{r}}^{\rm{N}}} \in {{\cal R}_{\rm{N}}}\left( {{\bf{p}},{\bf{\Theta }}} \right), r_k^{\rm{N}}\tau _{\rm{N}}^{{\rm{c}}} \ge sa_k^t, \forall k,\\
\label{C3-d}&\eqref{C1-d},\eqref{C1-e},\eqref{C1-f},\eqref{C1-g},\eqref{C2-e},
\end{align}
\end{subequations}
respectively, where ${{\bm{\tau }}^{\rm{F}}} = \left[ {\tau _{\rm{F}}^{\rm{c}},\tau _t^{{\rm{loc}}}} \right]$, ${{\bm{\tau }}^{\rm{N}}} = \left[ {\tau _{\rm{N}}^{\rm{c}},\tau _t^{{\rm{loc}}}} \right]$, and ${{\cal R}_{\rm{N}}}\left( {{\bf{p}},{\bf{\Theta }}} \right)$ is defined in \eqref{capacity_region}. Note that constraint \eqref{C2-d} indicates the requirement of bandwidth allocation coefficients.

Since the optimization variables are tightly coupled in the constraints and the device scheduling variables $\left\{ {a_k^t} \right\}$ are binary, problems \eqref{C1}, \eqref{C2}, and \eqref{C3} are all non-convex. Generally, there are no standard methods to solve them optimally. Although the above problems have similar forms, we propose different algorithms to obtain their high-quality solutions by exploiting their unique structures, as will be shown in Section IV.
\vspace{-10pt}
\section{Proposed Solution}
In this section, we investigate the joint communication-computation resource allocation and IRS phase-shifts design in IRS aided FL to minimize the per-round latency under the proposed three transmission protocols.
\vspace{-10pt}
\subsection{Proposed Solution to I-TDMA}
Defining $E_{k,{\rm{T}}}^{\rm{c}} = {p_k}\tau _k^{\rm{c}},\forall k \in {\cal K}$ and replacing ${p_k}$ with $E_{k,{\rm{T}}}^{\rm{c}}/\tau _k^{\rm{c}}$, problem \eqref{C1} can be equivalently rewritten as
\begin{subequations}\label{C4}
\begin{align}
\label{C4-a}\mathop {\min }\limits_{\left\{ {{{\bf{\Theta }}_k}} \right\},{\bf{a}},{{\bm{\tau }}^{\rm{T}}},\left\{ {E_{k,{\rm{T}}}^{\rm{c}}} \right\},{\bf{f}}} &\sum\nolimits_{k = 1}^K {\tau _k^{{\rm{c}}}}  + \tau _t^{{\rm{loc}}}\\
\label{C4-b}{\rm{s.t.}}\;\;\;\;\;\;\;\;\;\;\;&E_{k,{\rm{T}}}^{\rm{c}} + \xi {C_k}{D_k}f_k^2 \le E_k^{\max },~\forall k,\\
\label{C4-c}&B\tau _k^{\rm{c}}{\log _2}\left( {1 + \frac{{E_{k,{\rm{T}}}^{\rm{c}}{{\left| {{h_k}\left( {{{\bf{\Theta }}_k}} \right)} \right|}^2}}}{{\tau _k^{\rm{c}}B{\sigma ^2}}}} \right)\nonumber\\
&\ge sa_k^t, \forall k,\\
\label{C4-d}&\eqref{C1-d},\eqref{C1-e},\eqref{C1-f},\eqref{C1-g},\eqref{C1-h}.
\end{align}
\end{subequations}
For problem \eqref{C4}, we have the following proposition.
\begin{pos}
The optimal $\left\{ {{f_k},\tau _t^{{\rm{loc}}}} \right\}$ of problem \eqref{C4} is denoted by $\left\{ {f_k^*,{{\left( {\tau _t^{{\rm{loc}}}} \right)}^*}} \right\}$. Then, it follows that
\begin{align}\label{computation_frequency}
f_k^* = a_k^t\frac{{{C_k}{D_k}}}{{{{\left( {\tau _t^{{\rm{loc}}}} \right)}^*}}},\forall k \in {\cal K}.
\end{align}
\end{pos}
\begin{proof}
We show \eqref{computation_frequency} by contradiction. Suppose that $\left\{ {f_k^*} \right\}$ achieves the optimal solution of problem \eqref{C4} and there exists ${f_k^*}$ satisfying $f_k^* > a_k^t{C_k}{D_k}/{\left( {\tau _t^{{\rm{loc}}}} \right)^*}$, $\exists k \in {\cal K}$. Then, we construct a different solution $\left\{ {{{\tilde f}_k}} \right\}$ which satisfies ${{\tilde f}_k} = a_k^t{C_k}{D_k}/{\left( {\tau _t^{{\rm{loc}}}} \right)^*}$, $\forall k \in {\cal K}$. For device $k$ whose $f_k^* > a_k^t{C_k}{D_k}/{\left( {\tau _t^{{\rm{loc}}}} \right)^*}$, it is obvious that $\xi {C_k}{D_k}\tilde f_k^2 < \xi {C_k}{D_k}{\left( {f_k^*} \right)^2}$, which leads to $\tilde E_{k,{\rm{T}}}^{\rm{c}} > E_{k,{\rm{T}}}^{{\rm{c*}}}$. Note that $\left\{ {\tilde E_{k,{\rm{T}}}^{\rm{c}}} \right\}$ and $\left\{ {E_{k,{\rm{T}}}^{{\rm{c*}}}} \right\}$ are two sets of values of $E_{k,{\rm{T}}}^{\rm{c}}$ under the solutions of $\left\{ {{{\tilde f}_k}} \right\}$ and $\left\{ {f_k^*} \right\}$, respectively. From \eqref{C4-c}, it is noticed that $\tau _k^{\rm{c}}$ monotonically decreases with respect to $E_{k,{\rm{T}}}^{\rm{c}}$. Obviously, the objective value under $\left\{ {{{\tilde f}_k}} \right\}$ is smaller than that under $\left\{ {{{\tilde f}_k}} \right\}$, which contradicts the assumption that $\left\{ {f_k^*} \right\}$ is optimal. Hence, the optimal $\left\{ {{f_k}} \right\}$ satisfies \eqref{computation_frequency}.
\end{proof}

Proposition 1 reveals that the computational frequency at edge devices should be set in a way that all edge devices complete the local training concurrently. Based on Proposition 1, problem \eqref{C4} can be equivalently transformed to
\begin{subequations}\label{C5}
\begin{align}
\label{C5-a}\mathop {\min }\limits_{\left\{ {{{\bf{\Theta }}_k}} \right\},{\bf{a}},{{\bm{\tau }}^{\rm{T}}},\left\{ {E_{k,{\rm{T}}}^{\rm{c}}} \right\}} &\sum\nolimits_{k = 1}^K {\tau _k^{{\rm{c}}}}  + \tau _t^{{\rm{loc}}}\\
\label{C5-b}{\rm{s.t.}}\;\;\;\;\;\;\;\;\;&E_{k,{\rm{T}}}^{\rm{c}} + \xi a_k^t\frac{{{{\left( {{C_k}{D_k}} \right)}^3}}}{{{{\left( {\tau _t^{{\rm{loc}}}} \right)}^2}}} \le E_k^{\max },~\forall k,\\
\label{C5-c}&\eqref{C4-c},\eqref{C1-e},\eqref{C1-g},\eqref{C1-h}.
\end{align}
\end{subequations}
Although constraint \eqref{C4-c} of problem \eqref{C5} is non-convex due to the tightly coupled optimization variables $\left\{ {{{\bf{\Theta }}_k}} \right\}$, $\left\{ {\tau _k^{\rm{c}}} \right\}$, $\left\{ {E_{k,T}^{\rm{c}}} \right\}$, and $\left\{ {a_k^t} \right\}$, we observe that ${{{\bf{\Theta }}_k}}$'s are separate in each edge device's achievable rate. This observation suggests that ${{{\bf{\Theta }}_k}}$'s can be optimized by solving $K$ sub-problems in parallel. In particular, the optimal solution of ${{{\bf{\Theta }}_k}}$ is obtained by solving the optimization problem as follows
\begin{subequations}\label{C6}
\begin{align}
\label{C6-a}\mathop {\max }\limits_{{{{\bf{\Theta }}_k}}} \;\;&{\left| {{h_{d,k}} + {{\bf{g}}^H}{{\bf{\Theta }}_k}{{\bf{h}}_{r,k}}} \right|^2}\\
\label{C6-b}{\rm{s.t.}}\;\;&\left| {{{\left[ {{{\bf{\Theta }}_k}} \right]}_{n,n}}} \right| = 1,~\forall n.
\end{align}
\end{subequations}
As shown in \cite{wu2021intelligent}, the optimal solution of problem \eqref{C6}, denoted by ${\bf{\Theta }}_k^*$, can be derived as
\begin{align}\label{optimal_phaseshift}
{\left[ {{\bf{\Theta }}_k^*} \right]_{n,n}} = {e^{ - j\left( {\arg \left( {{h_{d,k}}} \right) + \arg \left( {{{\left[ {{\rm{diag}}\left( {{{\bf{g}}^H}} \right){{\bf{h}}_{r,k}}} \right]}_n}} \right)} \right)}},\forall n.
\end{align}
Define $\gamma _k^* = {\left| {{h_{d,k}} + {{\bf{g}}^H}{\bf{\Theta }}_k^*{{\bf{h}}_{r,k}}} \right|^2}$ and thereby problem \eqref{C5} can be rewritten as
\begin{subequations}\label{C7}
\begin{align}
\label{C7-a}\mathop {\min }\limits_{{\bf{a}},{{\bm{\tau }}^{\rm{T}}},\left\{ {E_{k,{\rm{T}}}^{\rm{c}}} \right\}} &\sum\nolimits_{k = 1}^K {\tau _k^{{\rm{c}}}}  + \tau _t^{{\rm{loc}}}\\
\label{C7-b}{\rm{s.t.}}\;\;\;\;\;&B\tau _k^{\rm{c}}{\log _2}\left( {1 + \frac{{E_{k,{\rm{T}}}^{\rm{c}}\gamma _k^*}}{{\tau _k^{\rm{c}}B{\sigma ^2}}}} \right) \ge sa_k^t,~\forall k,\\
\label{C7-c}&\eqref{C5-b},\eqref{C1-e},\eqref{C1-g}.
\end{align}
\end{subequations}
Problem \eqref{C7} is still challenging to be solved optimally due to a set of binary variables $\left\{ {a_k^t} \right\}$. To alleviate this burden, we propose an efficient method to obtain its optimal solution by exploiting its unique structure. Note that for the optimal solution of problem \eqref{C7}, constraint \eqref{C5-b} is met with equality, since otherwise we can always decrease the objective value by increasing $E_{k,{\rm{T}}}^{\rm{c}}$.

We first consider the case of the fixed value of ${\tau _t^{{\rm{loc}}}}$. For the given feasible value of ${\tau _t^{{\rm{loc}}}}$, denoted by $\bar \tau _t^{{\rm{loc}}}$, edge devices use up the whole available energy, which leads to
\begin{align}\label{communication_energy}
\bar E_{k,{\rm{T}}}^{\rm{c}} = a_k^t\left( {E_k^{\max } - \xi \frac{{{{\left( {{C_k}{D_k}} \right)}^3}}}{{{{\left( {\bar \tau _t^{{\rm{loc}}}} \right)}^2}}}} \right),\forall k \in {\cal K}.
\end{align}
In this case, problem \eqref{C7} can be rewritten in the following by dropping the constant term.
\begin{subequations}\label{C8}
\begin{align}
\label{C8-a}\mathop {\min }\limits_{{\bf{a}},{{\bm{\tau }}^{\rm{T}}}} \;&\sum\nolimits_{k = 1}^K {\tau _k^{{\rm{c}}}}\\
\label{C8-b}{\rm{s.t.}}\;&B\tau _k^{\rm{c}}{\log _2}\left( {1 + \frac{{\bar E_{k,{\rm{T}}}^{\rm{c}}\gamma _k^*}}{{\tau _k^{\rm{c}}B{\sigma ^2}}}} \right) \ge sa_k^t,~\forall k,\\
\label{C8-c}&\eqref{C1-e},\eqref{C1-g}.
\end{align}
\end{subequations}
To obtain the optimal $\left\{ {\tau _k^c} \right\}$ of problem \eqref{C8}, we provide the following proposition.
\begin{pos}
For problem \eqref{C8}, the optimal transmission time allocation, denoted by $\left\{ {{{\left( {\tau _k^c} \right)}^*}} \right\}$, is given by
\begin{align}\label{optimal_time_TDMA}
{\left( {\tau _k^{\rm{c}}} \right)^*} = a_k^t\frac{{ - s\bar E_{k,{\rm{T}}}^{\rm{c}}\gamma _k^*\ln 2}}{{B\left( {s{\sigma ^2}\ln 2 + {{\cal W}_{ - 1}}\left( {{\Upsilon _k}} \right)\bar E_{k,{\rm{T}}}^{\rm{c}}\gamma _k^*} \right)}},\forall k,
\end{align}
where
\begin{align}\label{term1}
{\Upsilon _k} =  - \frac{{s{\sigma ^2}\ln 2}}{{\bar E_{k,{\rm{T}}}^{\rm{c}}\gamma _k^*}}{e^{ - \frac{{s{\sigma ^2}}}{{\bar E_{k,{\rm{T}}}^{\rm{c}}\gamma _k^*}}\ln 2}},\forall k,
\end{align}
and ${{\cal W}_{ - 1}}\left( x \right)$ is the secondary branch of the Lambert W function.
\end{pos}
\begin{proof}
It is easy to show that
\begin{align}\label{function1}
h_k^{\rm{T}}\left( x \right) = Bx{\log _2}\left( {1 + \frac{{\bar E_{k,{\rm{T}}}^{\rm{c}}\gamma _k^*}}{{xB{\sigma ^2}}}} \right),\forall k,
\end{align}
is an increasing function with respect to $x$. Following that, constraint \eqref{C8-b} is met with equality at the optimal solution of problem \eqref{C8}. Hence, ${\left( {\tau _k^{\rm{c}}} \right)^*}$ is the unique solution of equation $h_k^{\rm{T}}\left( x \right) = s$. After some simple manipulations, $h_k^{\rm{T}}\left( x \right) = s$ is equivalent to
\begin{align}\label{temp}
\left( {\frac{{s{\sigma ^2}}}{{\bar E_{k,{\rm{T}}}^{\rm{c}}\gamma _k^*}} \!+\! \frac{{s\ln 2}}{{xB}}} \right){e^{ - \left( {\frac{{s{\sigma ^2}}}{{\bar E_{k,{\rm{T}}}^{\rm{c}}\gamma _k^*}} \!+\! \frac{{s\ln 2}}{{xB}}} \right)}} \!\!=\!\! \frac{{s{\sigma ^2}\ln 2}}{{\bar E_{k,{\rm{T}}}^{\rm{c}}\gamma _k^*}}{2^{ - \frac{{s{\sigma ^2}}}{{\bar E_{k,{\rm{T}}}^{\rm{c}}\gamma _k^*}}}},
\end{align}
which directly leads to the result in \eqref{optimal_time_TDMA}.
\end{proof}

For constraint \eqref{C1-e} in problem \eqref{C8}, it can be rewritten as
\begin{align}\label{weight_deivice}
\sum\nolimits_{k = 1}^K {\left( {1 - a_k^t} \right){D_k} \le } \rho ,
\end{align}
where $\rho  = \sqrt {\frac{{\kappa L{D^2}}}{{2\varepsilon }}}$ is a constant related to the target of the gradient error $\varepsilon$. By exploiting Proposition 2 and further relaxing the binary constraints in \eqref{C1-g}, problem \eqref{C8} is reduced to
\begin{subequations}\label{C9}
\begin{align}
\label{C9-a}\mathop {\min }\limits_{{\bf{a}}} \;&\sum\nolimits_{k = 1}^K {a_k^t\tilde \tau _k^{\rm{c}}} \\
\label{C9-b}{\rm{s.t.}}\;&\sum\nolimits_{k = 1}^K {\left( {1 - a_k^t} \right){D_k} \le } \rho\\
\label{C9-c}&0 \le a_k^t \le 1,~\forall k,
\end{align}
\end{subequations}
where $\tilde \tau _k^{\rm{c}} = {\left( {\tau _k^{\rm{c}}} \right)^*}/a_k^t$ and ${\left( {\tau _k^{\rm{c}}} \right)^*}$ is given in \eqref{optimal_time_TDMA}. Problem \eqref{C9} is convex since the objective function and all constraints are linear, which can be solved optimally by employing standard convex optimization techniques, such as CVX. Instead of using generic methods, we derive the closed-form expressions of the optimal $\left\{ {a_k^t} \right\}$ in the following proposition to gain useful insights.
\begin{pos}
The optimal $\left\{ {a_k^t} \right\}$ of problem \eqref{C9}, denoted by $\left\{ {{{\left( {a_k^t} \right)}^*}} \right\}$, can be obtained as
\begin{align}\label{optimal_schduling}
{{{\left( {a_k^t} \right)}^*}} = \begin{cases} 1, &I\left( {\tilde \tau _k^{\rm{c}},{D_k}} \right) \buildrel \Delta \over = \tilde \tau _k^{\rm{c}} - {\lambda ^*}{D_k} \le 0 \cr 0,
&{\rm{otherwise}}, \end{cases}
\end{align}
where ${\lambda ^*}$ is the optimal dual variable and can be obtained by sub-gradient method.
\end{pos}
\begin{proof}
Since problem \eqref{C9} is convex, it can be optimally solved by using the Lagrange duality method. To this end, the partial Lagrange function of problem \eqref{C9} can be written as
\begin{align}\label{Lagrange_function}
{\cal L}\left( {\left\{ {a_k^t} \right\},\lambda } \right) = \sum\limits_{k = 1}^K {a_k^t\tilde \tau _k^{\rm{c}}} \! +\! \lambda \left( {\sum\limits_{k = 1}^K {\left( {1 \!- \! a_k^t} \right){D_k} \!-\! \rho } } \right),
\end{align}
where $\lambda$ is the non-negative dual variable associated with constraint \eqref{C9-b}. Accordingly, the dual function is
\begin{align}\label{dual_function}
{\cal G}\left( \lambda  \right) = \mathop {\min }\limits_{\left\{ {a_k^t} \right\}} {\cal L}\left( {\left\{ {a_k^t} \right\},\lambda } \right),\;{\rm{s}}{\rm{.t}}{\rm{.}}\;\eqref{C9-c}.
\end{align}
Hence, the dual problem of problem \eqref{C9} is
\begin{align}\label{dual_problem}
\mathop {\max }\limits_{\lambda } \; {\cal G}\left( {\lambda } \right),\;{\rm{s}}{\rm{.t}}{\rm{.}}\;\lambda  \ge {\rm{0}}.
\end{align}
Note that problem \eqref{C9} can be solved by solving its dual problem \eqref{dual_problem} equivalently. In the following, we aim to obtain ${\cal G}\left( \lambda  \right)$ by solving problem \eqref{dual_function} optimally under the given $\lambda$. By dropping some constant terms, problem \eqref{dual_function} can be equivalently transformed to
\begin{align}\label{dual_function_problem}
\sum\nolimits_{k = 1}^K {a_k^t\left( {\tilde \tau _k^{\rm{c}} - \lambda {D_k}} \right)},\;{\rm{s}}{\rm{.t}}{\rm{.}}\;\eqref{C9-c}.
\end{align}
It is obvious that the ${a_k^t}$ should be set to 1 provided that $\tilde \tau _k^{\rm{c}} - \lambda {D_k} \le 0$ and should be set to 0 otherwise. Under the optimal dual variable ${\lambda ^*}$, the optimal scheduling variables are given in \eqref{optimal_schduling}, which thus completes the proof.
\end{proof}

Proposition 3 implies that the result in \eqref{optimal_schduling} admits a binary solution for device scheduling, which guarantees both the feasibility and optimality of the original problem \eqref{C8}. Note that $I\left( {\tilde \tau _k^{\rm{c}},{D_k}} \right)$ defined in \eqref{optimal_schduling} serves as a scheduling indicator for each edge device. It is evident that $I\left( {\tilde \tau _k^{\rm{c}},{D_k}} \right)$ decreases with respect to ${\left| {{h_k}\left( {{\bf{\Theta }}_k^*} \right)} \right|^2}$, ${{D_k}}$, and $\bar E_{k,T}^{\rm{c}}$, which makes sense to schedule a device with a higher channel gain, a large data-set, and sufficient communication energy.

Based on the above discussion, problem \eqref{C4} is solved optimally under the fixed value of $\tau _t^{{\rm{loc}}}$. Then, we demonstrate that the global optimal solution of problem \eqref{C4} can be obtained by performing a 1-D search over all candidate values of $\tau _t^{{\rm{loc}}}$. To this end, we derive both the lower bound and the upper bound of the optimal ${\left( {\tau _t^{{\rm{loc}}}} \right)^*}$ in the following lemma.
\begin{lem}
${\left( {\tau _t^{{\rm{loc}}}} \right)^*}$ satisfies $\tau _{{\rm{low}}}^{{\rm{loc}}} \le {\left( {\tau _t^{{\rm{loc}}}} \right)^*} \le \tau _{{\rm{up}}}^{{\rm{loc}}}$, where
\begin{align}\label{computation_time_lower}
\tau _{{\rm{low}}}^{{\rm{loc}}} = \mathop {\max }\limits_{k \in {\cal K}} \left\{ {{{\left( {\frac{{\xi C_k^3D_k^3}}{{E_k^{\max } - s{\sigma ^2}\ln 2/\gamma _k^*}}} \right)}^{1/2}}} \right\},
\end{align}
\begin{align}\label{computation_time_lower}
\tau _{{\rm{up}}}^{{\rm{loc}}} = \tau _{{\rm{low}}}^{{\rm{loc}}} + {\tau ^{c*}}\left( {\tau _{{\rm{low}}}^{{\rm{loc}}}} \right)
\end{align}
with ${\tau ^{c*}}\left( {\tau _{{\rm{low}}}^{{\rm{loc}}}} \right)$ denoting the optimal value of problem \eqref{C8} under $\tau _t^{{\rm{loc}}} = \tau _{{\rm{low}}}^{{\rm{loc}}}$.
\begin{proof}
Note that the condition
\begin{align}\label{energy_condition}
\bar E_{k,{\rm{T}}}^{\rm{c}} \ge \left( {s{\sigma ^2}\ln 2} \right)/\gamma _k^*,\forall k
\end{align}
should be satisfied to guarantee that equation $h_k^{\rm{T}}\left( x \right) = s$ has the unique solution, which leads to
\begin{align}\label{energy_condition1}
{\left( {\tau _t^{{\rm{loc}}}} \right)^*} \ge {\left( {\frac{{\xi C_k^3D_k^3}}{{E_k^{\max } - s{\sigma ^2}\ln 2/\gamma _k^*}}} \right)^{1/2}},\forall k.
\end{align}
Hence, we obtain its lower bound in \eqref{computation_time_lower}. Since $\tau _{{\rm{low}}}^{{\rm{loc}}}$ is one feasible solution of ${\tau _t^{{\rm{loc}}}}$ to problem \eqref{C4}, we have
\begin{align}\label{inequality}
{\left( {\tau _t^{{\rm{loc}}}} \right)^*} \!<\! {\left( {\tau _t^{{\rm{loc}}}} \right)^*} \!+\! {\tau ^{c*}}\left( {{{\left( {\tau _t^{{\rm{loc}}}} \right)}^*}} \right) \!\le \! \tau _{{\rm{low}}}^{{\rm{loc}}} \!+\! {\tau ^{c*}}\left( {\tau _{{\rm{low}}}^{{\rm{loc}}}} \right),
\end{align}
which thus completes the proof.
\end{proof}
\end{lem}

Lemma 2 implies that the optimal ${\left( {\tau _t^{{\rm{loc}}}} \right)^*}$ can be obtained by performing a 1-D exhaustive search over the region $\left[ {\tau _{{\rm{low}}}^{{\rm{loc}}},\tau _{{\rm{up}}}^{{\rm{loc}}}} \right]$. Note that problem \eqref{C7} can be solved optimally for $\forall \tau _t^{{\rm{loc}}} \in \left[ {\tau _{{\rm{low}}}^{{\rm{loc}}},\tau _{{\rm{up}}}^{{\rm{loc}}}} \right]$. For the fixed value of ${\tau _t^{{\rm{loc}}}}$, the computational complexity to obtain $\left\{ {{\bf{\Theta }}_k^*,{{\left( {E_{k,T}^{\rm{c}}} \right)}^*},{{\left( {a_k^t} \right)}^*},{{\left( {\tau _k^{\rm{c}}} \right)}^*}} \right\}$ is ${\cal{O}}\left( {N + K} \right)$. Therefore, the overall complexity for optimally solving the original problem \eqref{C1} is ${\cal{O}}\left( {\left( {N + K} \right)\left( {\tau _{{\rm{up}}}^{{\rm{loc}}} - \tau _{{\rm{low}}}^{{\rm{loc}}}} \right)/\varsigma } \right)$, where $\varsigma $ is a positive constant that controls the accuracy of a 1-D search.

\vspace{-10pt}
\subsection{Proposed Solution to I-FDMA}
It is evident that Proposition 1 is also applicable to I-FDMA. By exploiting Proposition 1, problem \eqref{C2} can be equivalently written as
\begin{subequations}\label{C10}
\begin{align}
\label{C10-a}\mathop {\min }\limits_{{\bf{\Theta }},{\bf{a}},{\bf{b}},{{\bm{\tau }}^{\rm{F}}},{\bf{p}}} \;&\tau _{\rm{F}}^{\rm{c}} + \tau _t^{{\rm{loc}}}\\
\label{C10-b}{\rm{s.t.}}\;\;\;\;\;\;&{p_k}\tau _{\rm{F}}^{\rm{c}} + \xi \frac{{{{\left( {{C_k}{D_k}} \right)}^3}}}{{{{\left( {\tau _t^{{\rm{loc}}}} \right)}^2}}} \le E_k^{\max },~\forall k,\\
\label{C10-c}&\eqref{C1-e},\eqref{C1-g}, \eqref{C2-c}, \eqref{C2-d}, \eqref{C2-e}.
\end{align}
\end{subequations}
For the I-FDMA case, the corresponding problem \eqref{C10} is more challenging to be solved than problem \eqref{C1} since a common IRS phase-shift matrix couples in rate constraints of all edge devices in \eqref{C2-c}. To solve problem \eqref{C10}, a decoupled optimization method is proposed, where we first jointly optimize IRS phase-shift and resource allocation under the given set of scheduling devices ${\cal K}_a^t$. On top of that, we employ a coordinate descent approach that optimizes the device scheduling variables $\left\{ {a_k^t} \right\}$.

\subsubsection{Optimization of $\left\{ {{\bf{\Theta }},{\bf{b}},{{\bf{\tau }}^{\rm{F}}},{\bf{p}}} \right\}$ under Given ${\cal K}_a^t$}
Define ${\bf{v}} = {\left[ {{e^{j{\theta _1}}}, \ldots ,{e^{j{\theta _N}}}} \right]^T}$, ${\bf{f}}_k^H = {{\bf{g}}^H}{\rm{diag}}\left( {{{\bf{h}}_{r,k}}} \right)$, $E_{{\rm{F,}}k}^{\rm{c}} = {p_k}\tau _{\rm{F}}^{\rm{c}}$, and ${e_k} = {b_k}\tau _{\rm{F}}^{\rm{c}}$. Under the given ${\cal K}_a^t$, problem \eqref{C10} can be equivalently transformed into
\begin{subequations}\label{C11}
\begin{align}
\label{C11-a}\mathop {\min }\limits_{{\bf{v}},\left\{ {{e_k}} \right\},{{\bm{\tau }}^{\rm{F}}},\left\{ {E_{{\rm{F,}}k}^{\rm{c}}} \right\}} \;&\tau _{\rm{F}}^{\rm{c}} + \tau _t^{{\rm{loc}}}\\
\label{C11-b}{\rm{s.t.}}\;\;\;\;\;\;\;\;\;&E_{{\rm{F,}}k}^{\rm{c}} \!+\! \xi \frac{{{{\left( {{C_k}{D_k}} \right)}^3}}}{{{{\left( {\tau _t^{{\rm{loc}}}} \right)}^2}}} \!\le\! E_k^{\max },\forall k \in {\cal K}_a^t,\\
\label{C11-c}&B{e_k}{\log _2}\left( {1 + \frac{{E_{{\rm{F,}}k}^{\rm{c}}{{\left| {{h_{d,k}} + {\bf{f}}_k^H{\bf{v}}} \right|}^2}}}{{{e_k}B{\sigma ^2}}}} \right)\nonumber\\
&\ge s,\forall k \in {\cal K}_a^t,\\
\label{C11-d}&\sum\nolimits_{k \in {\cal K}_a^t} {{e_k}}  \le \tau _{\rm{F}}^{\rm{c}},{e_k} \ge 0,\forall k \in {\cal K}_a^t.\\
\label{C11-e}&\left| {{{\left[ {\bf{v}} \right]}_n}} \right| = 1,\forall n.
\end{align}
\end{subequations}
To deal with the non-convex constraint \eqref{C11-c}, a set of slack variables $\left\{ {{Y_k}} \right\}$ are introduced and thereby problem \eqref{C11} can be reformulated as follows
\begin{subequations}\label{C12}
\begin{align}
\label{C12-a}\mathop {\min }\limits_{\left\{ {{\bf{v}},{e_k},{{\bm{\tau }}^{\rm{F}}},E_{{\rm{F}},k}^{\rm{c}},{Y_k}} \right\}} &\tau _{\rm{F}}^{\rm{c}} + \tau _t^{{\rm{loc}}}\\
\label{C12-b}{\rm{s.t.}}\;\;\;\;\;\;\;\;\;&B{e_k}{\log _2}\left( {1 \!\!+\!\! \frac{{{Y_k}}}{{{e_k}B{\sigma ^2}}}} \right) \!\ge\! s,\forall k \in {\cal K}_a^t,\\
\label{C12-c}&{Y_k} \le E_{{\rm{F,}}k}^{\rm{c}}{\left| {{h_{d,k}} + {\bf{f}}_k^H{\bf{v}}} \right|^2},\forall k \in {\cal K}_a^t,\\
\label{C12-d}&\eqref{C11-b},\eqref{C11-d},\eqref{C11-e}.
\end{align}
\end{subequations}
Note that constraint \eqref{C12-c} is met with equality at the optimal solution of problem \eqref{C12}, since otherwise the objective value can be always decreased by increasing ${{Y_k}}$ until \eqref{C12-c} becomes active. However, constraints \eqref{C12-c} and \eqref{C11-e} are still non-convex. Nevertheless, we observe that ${{{{\left| {{h_{d,k}} + {\bf{f}}_k^H{\bf{v}}} \right|}^2}} \mathord{\left/
 {\vphantom {{{{\left| {{h_{d,k}} + {\bf{f}}_k^H{\bf{v}}} \right|}^2}} {\frac{1}{{E_{{\rm{F,}}k}^{\rm{c}}}}}}} \right.
 \kern-\nulldelimiterspace} {\frac{1}{{E_{{\rm{F,}}k}^{\rm{c}}}}}}$ is jointly convex with respect to ${\bf{v}}$ and $1/E_{{\rm{F,}}k}^{\rm{c}}$, which motivates us to employ the SCA technique to address this issue. At an arbitrarily given point $\left\{ {{\bf{\bar v}},\bar E_{{\rm{F,}}k}^{\rm{c}}} \right\}$, we adopt first-order Taylor expansion to obtain a lower bound of $E_{{\rm{F,}}k}^{\rm{c}}{\left| {{h_{d,k}} + {\bf{f}}_k^H{\bf{v}}} \right|^2}$ as follows
 \begin{align}\label{SCA_bound}
E_{{\rm{F}},k}^{\rm{c}}{\left| {{h_{d,k}} + {\bf{f}}_k^H{\bf{v}}} \right|^2} \ge&  - \frac{1}{{E_{{\rm{F}},k}^{\rm{c}}}}{\left( {\bar E_{{\rm{F}},k}^{\rm{c}}} \right)^2}{\left| {{h_{d,k}} + {\bf{f}}_k^H{\bf{\bar v}}} \right|^2}\nonumber\\
&+ 2\bar E_{{\rm{F}},k}^{\rm{c}}{\mathop{\rm Re}\nolimits} \left( {\left( {h_{d,k}^H{\bf{f}}_k^H + {{{\bf{\bar v}}}^H}{{\bf{f}}_k}{\bf{f}}_k^H} \right){\bf{v}}} \right)\nonumber\\
& + 2\bar E_{{\rm{F}},k}^{\rm{c}}\left( {{{\left| {{h_{d,k}}} \right|}^2} + {\mathop{\rm Re}\nolimits} \left( {\bar E_{{\rm{F}},k}^{\rm{c}}h_{d,k}^H{\bf{f}}_k^H{\bf{\bar v}}} \right)} \right)\nonumber\\
&\buildrel \Delta \over = g_k^{{\rm{lb}}}\left( {{\bf{v}},E_{{\rm{F}},k}^{\rm{c}}} \right).
\end{align}
Note that $g_k^{{\rm{lb}}}\left( {{\bf{v}},E_{{\rm{F}},k}^{\rm{c}}} \right)$ is jointly concave with respect to ${\bf{v}}$ and ${E_{{\rm{F}},k}^{\rm{c}}}$. As such, we transform constraint \eqref{C12-c} as
\begin{align}\label{new_constraint}
{Y_k} \le g_k^{{\rm{lb}}}\left( {{\bf{v}},E_{{\rm{F}},k}^{\rm{c}}} \right),\forall k \in {\cal K}_a^t,
\end{align}
which is a convex constraint. The remaining challenging for solving problem \eqref{C12} is constraint \eqref{C11-e}. To address this issue, we relax the unit-modules constraints as
\begin{align}\label{relaxed_module}
\left| {{{\left[ {\bf{v}} \right]}_n}} \right| \le 1,\forall n.
\end{align}
Then, we approximate problem \eqref{C12} as
\begin{subequations}\label{C13}
\begin{align}
\label{C13-a}\mathop {\min }\limits_{\left\{ {{\bf{v}},{e_k},{{\bm{\tau }}^{\rm{F}}},E_{{\rm{F}},k}^{\rm{c}},{Y_k}} \right\}} &\tau _{\rm{F}}^{\rm{c}} + \tau _t^{{\rm{loc}}}\\
\label{C13-b}{\rm{s.t.}}\;\;\;\;\;\;\;\;\;&\eqref{C12-b}, \eqref{new_constraint},\eqref{C11-b},\eqref{C11-d},\eqref{relaxed_module}.
\end{align}
\end{subequations}
Problem \eqref{C13} is convex and thereby it can be solved successively by using CVX until the convergence is achieved. Note that the converged solution, denoted by ${{\bf{v}}^*}$, may not satisfy constraint \eqref{C11-e}. For this case, one feasible solution of ${\bf{v}}$ satisfying \eqref{C11-e} can be constructed as
\begin{align}\label{phase_reconstruct}
{\left[ {{\bf{\tilde v}}} \right]_n} = {\left[ {{{\bf{v}}^*}} \right]_n}/\left| {{{\left[ {{{\bf{v}}^*}} \right]}_n}} \right|,\forall n.
\end{align}
Under the obtained ${{\bf{\tilde v}}}$, we obtain the solution of the remaining optimization variables by using the proposed method. The computational complexity of the proposed algorithm is ${\cal{O}}\left( {{I_{{\rm{iter}}}}{{\left( {N + 3K} \right)}^{3.5}}} \right)$, where ${{I_{{\rm{iter}}}}}$ denotes the number of iterations to perform SCA.

\subsubsection{Optimization of device scheduling}
Then, we propose a coordinate descent method to optimize device scheduling variables $\left\{ {a_k^t} \right\}$. Recall that ${\bf{a}} = \left[ {a_1^t, \ldots ,a_K^t} \right]$. Based on the coordinate descent, the direction of only one variable $a_k^t$ is updated successively. Specifically, with an initial ${{\bf{a}}^{\left( 0 \right)}} = \left[ {1, \ldots ,1} \right]$, we denote ${{\bf{a}}^{\left( {l - 1} \right)}} = \left[ {a_1^{t,\left( {l - 1} \right)}, \ldots ,a_K^{t,\left( {l - 1} \right)}} \right]$ as the device scheduling vector in the ${\left( {l - 1} \right)}$-th iteration, $l \in \left\{ {1,2, \ldots } \right\}$. Accordingly, we denote ${\tau _{\rm{F}}}\left( {{{\bf{a}}^{\left( {l - 1} \right)}}} \right)$ as the objective value of problem \eqref{C10} under ${\bf{a}} = {{\bf{a}}^{\left( {l - 1} \right)}}$, which can be obtained based on our proposed algorithm. Let
\begin{align}\label{reward}
\Delta \tau _{{\rm{F}},k}^{\left( l \right)} = {\tau _{\rm{F}}}\left( {{{\bf{a}}^{\left( {l - 1} \right)}}\left( k \right)} \right) - {\tau _{\rm{F}}}\left( {{{\bf{a}}^{\left( {l - 1} \right)}}} \right),\forall k,
\end{align}
where
\begin{align}\label{swapping}
{{\bf{a}}^{\left( {l - 1} \right)}}\left( k \right) = \left[ {a_1^{t,\left( {l - 1} \right)}, \ldots ,1 - a_k^{t,\left( {l - 1} \right)}, \ldots ,a_K^{t,\left( {l - 1} \right)}} \right].
\end{align}
In the $l$-th iteration, the scheduling vector is updated as ${{\bf{a}}^{\left( l \right)}} = {{\bf{a}}^{\left( {l - 1} \right)}}\left( {k_l^*} \right)$ with $k_l^* = {{\mathop{\rm argmin}\nolimits} _{k \in {\cal K}}}\Delta \tau _{{\rm{F}},k}^{\left( l \right)}$. The iterations proceed until constraint \eqref{C1-e} becomes infeasible.
\vspace{-10pt}
\subsection{Proposed Solution to I-NOMA}
By exploiting Proposition 1, problem \eqref{C3} of the I-NOMA case can be equivalently reformulated as
\begin{subequations}\label{C14}
\begin{align}
\label{C14-a}\mathop {\min }\limits_{{\bf{\Theta }},{\bf{a}},{{\bm{\tau }}^{\rm{N}}},{\bf{p}}} \;&\tau _{\rm{N}}^{\rm{c}} + \tau _t^{{\rm{loc}}}\\
\label{C14-b}{\rm{s.t.}}\;\;\;\;\;\;&{p_k}\tau _{\rm{N}}^{\rm{c}} + \xi \frac{{{{\left( {{C_k}{D_k}} \right)}^3}}}{{{{\left( {\tau _t^{{\rm{loc}}}} \right)}^2}}} \le E_k^{\max },~\forall k,\\
\label{C14-c}&\eqref{C1-e},\eqref{C1-g}, \eqref{C2-e},\eqref{C3-c}.
\end{align}
\end{subequations}
Similar to Section III-B, we exploit the coordinate descent method to decouple the device scheduling and resource allocation. First, we focus on the case of the joint optimization of $\left\{ {{\bf{\Theta }},{{\bm{\tau }}^{\rm{N}}},{\bf{p}}} \right\}$ under the given set of scheduled devices ${\cal K}_a^t$. Let $E_{{\rm{N}},k}^c = {p_k}\tau _{\rm{N}}^{\rm{c}}$ and then constraint \eqref{C3-c} becomes
\begin{align}\label{capacity_region2}
{{\bf{r}}^{\rm{N}}} \in {{\tilde {\cal R}}_{\rm{N}}}\left( {\left\{ {E_{{\rm{N}},k}^c} \right\},{\bf{\Theta }},\tau _{\rm{N}}^{\rm{c}},{\cal K}_a^t} \right),r_k^{\rm{N}}\tau _{\rm{N}}^{\rm{c}} \ge s,\forall k \in {\cal K}_a^t,
\end{align}
where
\begin{align}\label{capacity_region3}
\begin{array}{l}
{{\tilde {\cal R}}_{\rm{N}}}\left( {\left\{ {E_{{\rm{N}},k}^c} \right\},{\bf{\Theta }},\tau _{\rm{N}}^{\rm{c}},{\cal K}_a^t} \right)\\
 = \left\{ {\begin{array}{*{20}{l}}
{{{\bf{r}}^{\rm{N}}} \in \mathbb{R}_ + ^{\left| {{\cal K}_a^t} \right| \times 1}:\sum\nolimits_{k \in {\cal J}} {r_k^N} }\\
{ \le B{\tau _{\rm{N}}^{\rm{c}}}{{\log }_2}\left( {1 \!\!+\!\! \frac{{\sum\nolimits_{k \in {\cal J}} {E_{{\rm{N}},k}^c{{\left| {{h_k}\left( {\bf{\Theta }} \right)} \right|}^2}} }}{{B\tau _{\rm{N}}^{\rm{c}}{\sigma ^2}}}} \right),\forall {\cal J} \subseteq {\cal K}_a^t}
\end{array}} \right\},
\end{array}
\end{align}

The corresponding optimization problem under the given ${\cal K}_a^t$ can be written as
\begin{subequations}\label{C15}
\begin{align}
\label{C15-a}\mathop {\min }\limits_{{\bf{\Theta }},{{\bm{\tau }}^{\rm{N}}},{\left\{ {E_{{\rm{N}},k}^c} \right\}}} \;&\tau _{\rm{N}}^{\rm{c}} + \tau _t^{{\rm{loc}}}\\
\label{C15-b}{\rm{s.t.}}\;\;\;\;\;\;&E_{{\rm{N}},k}^c + \xi \frac{{{{\left( {{C_k}{D_k}} \right)}^3}}}{{{{\left( {\tau _t^{{\rm{loc}}}} \right)}^2}}} \le E_k^{\max },~\forall k,\\
\label{C15-c}&\eqref{C2-e},\eqref{capacity_region2}.
\end{align}
\end{subequations}
For problem \eqref{C15}, \eqref{capacity_region2} is intractable since the number of constraints in \eqref{capacity_region2} is ${2^{\left| {{\cal K}_a^t} \right|}} - 1$.  Without loss of generality, we assume that ${S_1} \le  \ldots  \le {S_{\left| {{\cal K}_a^t} \right|}}$ , where ${S_k} = E_{{\rm{N}},k}^c{\left| {{h_k}\left( {\bf{\Theta }} \right)} \right|^2}$. Then, it is observed that constraint \eqref{capacity_region2} is equivalent to
\begin{align}\label{capacity_region4}
B\tau _{\rm{N}}^{\rm{c}}{\log _2}\left( {1 + \frac{{\sum\nolimits_{k = 1}^m {{S_k}} }}{{B\tau _{\rm{N}}^{\rm{c}}{\sigma ^2}}}} \right) \ge ms, m = 1, \ldots ,\left| {{\cal K}_a^t} \right|.
\end{align}
Compared to \eqref{capacity_region2}, the number of constraints in \eqref{capacity_region4} is significantly reduced from  ${2^{\left| {{\cal K}_a^t} \right|}} - 1$ to ${\left| {{\cal K}_a^t} \right|}$. Under the fixed value of ${\tau _t^{{\rm{loc}}}}$, denoted by $\bar \tau _t^{{\rm{loc}}}$, problem \eqref{C15} is reduced to
\begin{subequations}\label{C16}
\begin{align}
\label{C16-a}\mathop {\min }\limits_{{\bf{\Theta }},\tau _{\rm{N}}^{\rm{c}}} \;&\tau _{\rm{N}}^{\rm{c}}\\
\label{C16-b}{\rm{s.t.}}\;\;&B\tau _{\rm{N}}^{\rm{c}}{\log _2}\left( {1 \!\!+\!\! \frac{{\sum\nolimits_{k = 1}^m {{{\bar S}_k}} }}{{B\tau _{\rm{N}}^{\rm{c}}{\sigma ^2}}}} \right) \!\!\ge\!\! ms, m = 1, \ldots ,\left| {{\cal K}_a^t} \right|\\
\label{C16-c}&\eqref{C2-e},
\end{align}
\end{subequations}
where ${{\bar S}_k} = \bar E_{{\rm{N}},k}^c{\left| {{h_k}\left( {\bf{\Theta }} \right)} \right|^2}$ with $\bar E_{{\rm{N}},k}^c = E_k^{\max } - \xi \frac{{{{\left( {{C_k}{D_k}} \right)}^3}}}{{{{\left( {\bar \tau _t^{{\rm{loc}}}} \right)}^2}}}$. It is noticed that $\tau _{\rm{N}}^{\rm{c}}$ decreases with respect to ${\sum\nolimits_{k = 1}^m {{{\bar S}_k}} }$, which motivates us to maximize $\sum\nolimits_{k = 1}^{\left| {{\cal K}_a^t} \right|} {{{\bar S}_k}}$ by optimizing ${\bf{\Theta }}$. Hence, we consider the following optimization problem
\begin{subequations}\label{C17}
\begin{align}
\label{C17-a}\mathop {\max }\limits_{\bf{v}} \;\;&\sum\nolimits_{k = 1}^{\left| {{\cal K}_a^t} \right|} {\bar E_{{\rm{N}},k}^c{{\left| {{h_{d,k}} + {\bf{f}}_k^H{\bf{v}}} \right|}^2}}\\
\label{C17-b}{\rm{s.t.}}\;\;&\left| {{{\left[ {\bf{v}} \right]}_{n,n}}} \right| = 1,~\forall n,
\end{align}
\end{subequations}
where ${\bf{v}} = {\left[ {{e^{j{\theta _1}}}, \ldots ,{e^{j{\theta _N}}}} \right]^T}$ and ${\bf{f}}_k^H = {{\bf{g}}^H}{\rm{diag}}\left( {{{\bf{h}}_{r,k}}} \right)$. Problem \eqref{C17} is a non-convex one due to the non-concave objective function and the unit-modules constraints of IRS phase-shifts. Nevertheless, the convexity of the objective function \eqref{C17-a} facilitates us to employ the SCA technique to solve it. By taking the first-order Taylor expansion of \eqref{C17-a} at arbitrarily feasible point ${{\bf{\bar v}}}$, we obtain a lower bound of \eqref{C17-a} as
\begin{align}\label{channel_gain_lower_bound}
2{\mathop{\rm Re}\nolimits} \left( {\left( {\sum\nolimits_{k = 1}^{\left| {{\cal K}_a^t} \right|} {\bar E_{{\rm{N}},k}^c{{{\bf{\bar v}}}^H}{{\bf{f}}_k}{\bf{f}}_k^H + h_{d,k}^H{\bf{f}}_k^H} } \right){\bf{v}}} \right) + Z,
\end{align}
where $Z = \sum\nolimits_{k = 1}^{\left| {{\cal K}_a^t} \right|} {\bar E_{{\rm{N}},k}^c{Z_k}}$ is a constant with
\begin{align}\label{Z_k}
{Z_k} =  - {\left| {{h_{d,k}} \!+\! {\bf{f}}_k^H{\bf{\bar v}}} \right|^2} \!+\! 2{\mathop{\rm Re}\nolimits} \left( {{{\left( {{h_{d,k}} \!+\! {\bf{f}}_k^H{\bf{\bar v}}} \right)}^H}{h_{d,k}}} \right).
\end{align}
The optimal solution of ${\bf{v}}$ for maximizing \eqref{channel_gain_lower_bound} is derived as
\begin{align}\label{local_optimal_v}
{\bf{\hat v}} = \exp \left( {j\arg \left( {\sum\nolimits_{k = 1}^{\left| {K_a^t} \right|} {\bar E_{{\rm{N}},k}^c\left( {{{\bf{f}}_k}{\bf{f}}_k^H{\bf{\bar v}} + {{\bf{f}}_k}{h_{d,k}}} \right)} } \right)} \right).
\end{align}
Based on \eqref{local_optimal_v}, one local optimal solution of \eqref{C17}, denoted by ${{\bf{v}}^\star}$, can be obtained by maximizing \eqref{channel_gain_lower_bound} iteratively  until convergence is achieved.

Let $\bar S_k^\star = \bar E_{{\rm{N}},k}^c{\left| {{h_{d,k}} + {\bf{f}}_k^H{{\bf{v}}^\star}} \right|^2}$ and $\bar S_{\pi \left( 1 \right)}^ \star  \le  \ldots  \le \bar S_{\pi \left( K \right)}^ \star$. Then, the optimal $\tau _{\rm{N}}^c$ under the obtained $\bar S_k^\star$ can be derived as
\begin{align}\label{optimal_NOMA_time}
{\left( {\tau _{\rm{N}}^c} \right)^ \star } \!\!=\!\! \mathop {\max }\limits_{m \in \left\{ {1, \ldots \left| {K_a^t} \right|} \right\}} \left( { - \frac{{ms\sum\nolimits_{k = 1}^m {\bar S_{\pi \left( k \right)}^ \star } }}{{B\left( {ms{\sigma ^2} \!\!+\!\! \frac{{{W_{ - 1}}\left( {{X_m}} \right)\sum\nolimits_{k = 1}^m {\bar S_{\pi \left( k \right)}^ \star } }}{{\ln 2}}} \right)}}} \right),
\end{align}
where ${X_m} =  - \frac{{{2^{ - \frac{{ms{\sigma ^2}}}{{\sum\nolimits_{k = 1}^m {\bar S_{\pi \left( k \right)}^ \star } }}}}ms{\sigma ^2}\ln 2}}{{\sum\nolimits_{k = 1}^m {\bar S_{\pi \left( k \right)}^ \star } }}$.
Similar to the discussion in Section IV-A, the optimized value of $\tau _t^{{\rm{loc}}}$ to problem \eqref{C15} can be obtained by a 1-D search. In the outer layer, the device scheduling variable $\left\{ {a_k^t} \right\}$ is updated via the coordinate descent method as discussed in Section IV-B. These details are omitted due to its brevity.

\vspace{-12pt}
\subsection{Further Discussion}
Note that the proposed optimization framework for latency minimization problems is also applicable to the objective for minimizing the bound of optimality gap. Recall from Theorem 1 that decreasing the upper bound of the optimality gap in \eqref{optimality_gap} increases with respect to ${A_t}$ in \eqref{e_t_bound}. Taking the case of I-TDMA as an example, the corresponding optimization problem to minimize ${A_t}$ under the latency constraint is formulated as
\begin{subequations}\label{C18}
\begin{align}
\label{C18-a}\mathop {\min }\limits_{\left\{ {{{\bf{\Theta }}_k}} \right\},{\bf{a}},{{\bm{\tau }}^{\rm{T}}},{\bf{p}},{\bf{f}}} &\frac{{2\varepsilon }}{{L{D^2}}}{\left( {\sum\nolimits_{k = 1}^K {\left( {1 - a_k^t} \right){D_k}} } \right)^2}\\
\label{C18-b}{\rm{s.t.}}\;\;\;\;\;\;&\sum\nolimits_{k = 1}^K {\tau _k^{\rm{c}}}  + \tau _t^{{\rm{loc}}} \le \bar \tau,\\
\label{C18-c}&\eqref{C1-b}, \eqref{C1-c},\eqref{C1-d},\eqref{C1-f}, \eqref{C1-g},\eqref{C1-h}.
\end{align}
\end{subequations}
Under the given target value of $\kappa$, we denote the optimal objective value of problem \eqref{C1} as $\tau _{\rm{T}}^*\left( \kappa  \right)$, which can be easily obtained by using the algorithm proposed in Section IV-A. It is evident that $\tau _{\rm{T}}^*\left( \kappa  \right)$ is non-increasing with respect to $\kappa $. Hence, the optimal value of problem \eqref{C18} is the minimum of $\kappa $ satisfying $\tau _{\rm{T}}^*\left( \kappa  \right) \le \bar \tau$, denoted by $A_t^*$, which can be efficiently obtained via a bisection search.

It is observed from \eqref{C18-a} that $A_t^*$ is lower bounded by $0$, which can be achieved by scheduling all edge devices. To shed light on the impact of the IRS on the device scheduling, we unveil a sufficient condition for achieving a full scheduling under the specific latency requirement.
\begin{pos}
The optimal device scheduling of problem \eqref{C18} is $a_k^t = 1,\forall k$, provided that
\begin{align}\label{full_scheduling_condition}
E_k^{\max } \!>\! \frac{{\xi {{\left( {{C_k}{D_k}} \right)}^3}}}{{{{\bar \tau }^2}}}, N \!\ge\! \sqrt {\frac{{\left( {{2^{sK/\left( {B\bar \tau _{\rm{T}}^{\rm{c}}} \right)}} - 1} \right)B\bar \tau _{\rm{T}}^{\rm{c}}{\sigma ^2}}}{{K\bar E_k^{\rm{c}}\rho _r^2\rho _g^2}}},
\end{align}
\end{pos}
where ${\rho _r}{\rho _g} = \mathop {\min }\limits_{k \in {\cal K},n \in {\cal N}} \left| {{{\left[ {{{\bf{h}}_{r,k}}} \right]}_n}{{\left[ {\bf{g}} \right]}_n}} \right|$, $\bar \tau _{\rm{T}}^{\rm{c}} = \bar \tau  - \bar \tau _t^{{\rm{loc}}}$, $\bar E_k^{\rm{c}} = E_k^{\max } - \xi {\left( {{C_k}{D_k}} \right)^3}/{\left( {\bar \tau _t^{{\rm{loc}}}} \right)^2}$ with ${\bar \tau _t^{{\rm{loc}}}}$ denoting any feasible value of $\tau _t^{{\rm{loc}}}$ satisfying $E_k^{\max } > \xi {\left( {{C_k}{D_k}} \right)^3}/{\left( {\tau _t^{{\rm{loc}}}} \right)^2}$.
\begin{proof}
Under the condition that $E_k^{\max } \!\!>\!\! \xi {\left( {{C_k}{D_k}} \right)^3}/{{\bar \tau }^2}$,$\forall k$, there always exists $\bar \tau _t^{{\rm{loc}}} \in \left( {0,\bar \tau } \right)$, which satisfies $E_k^{\max } > \xi {\left( {{C_k}{D_k}} \right)^3}/{\left( {\bar \tau _t^{{\rm{loc}}}} \right)^2},\forall k$. Let $\bar E_k^{\rm{c}} = E_k^{\max } - \xi {\left( {{C_k}{D_k}} \right)^3}/{\left( {\bar \tau _t^{{\rm{loc}}}} \right)^2}$ and $\bar \tau _k^{\rm{c}} = \left( {\bar \tau  - \bar \tau _t^{{\rm{loc}}}} \right)/K,\forall k$. Then, the transmit power of device $k$ is given by ${p_k} = \bar E_k^{\rm{c}}/\bar \tau _k^{\rm{c}}$ and thereby the number of uploaded bits of device $k$ can be derived as
\begin{align}\label{uploaded bits}
r_k^{\rm{T}}\tau _k^{\rm{c}} &= \frac{{B\bar \tau _{\rm{T}}^{\rm{c}}}}{K}{\log _2}\left( {1 + \frac{{\bar E_k^{\rm{c}}K{{\left| {{h_{d,k}} + {{\bf{g}}^H}{\bf{\Theta }}_k^*h_{r,k}^H} \right|}^2}}}{{\bar \tau _{\rm{T}}^{\rm{c}}B{\sigma ^2}}}} \right)\nonumber\\
& \ge \frac{{B\bar \tau _{\rm{T}}^{\rm{c}}}}{K}{\log _2}\left( {1 + \frac{{\bar E_k^{\rm{c}}K{{\left| {{{\bf{g}}^H}{\bf{\Theta }}_k^*h_{r,k}^H} \right|}^2}}}{{\bar \tau _{\rm{T}}^{\rm{c}}B{\sigma ^2}}}} \right)\nonumber\\
&\ge \frac{{B\bar \tau _{\rm{T}}^{\rm{c}}}}{K}{\log _2}\left( {1 + \frac{{\bar E_k^{\rm{c}}K\rho _r^2\rho _g^2{N^2}}}{{\bar \tau _{\rm{T}}^{\rm{c}}B{\sigma ^2}}}} \right).
\end{align}
According to \eqref{uploaded bits}, $r_k^{\rm{T}}\tau _k^{\rm{c}} \ge s,\forall k$ always holds when \eqref{full_scheduling_condition} is satisfied, which implies that $a_k^t = 1,\forall k$ is feasible for problem \eqref{C18}. Thus, the proof is completed.
\end{proof}

In Proposition 4, $\xi {\left( {{C_k}{D_k}} \right)^3}/{{\bar \tau }^2}$ represents the minimum energy required for the local model training at device $k$. Under the condition that $E_k^{\max } > \xi {\left( {{C_k}{D_k}} \right)^3}/{{\bar \tau }^2}$, Proposition 4 explicitly answers the question on how many IRS elements are needed to support a full device scheduling. When \eqref{full_scheduling_condition} is satisfied, the optimality gap in \eqref{optimality_gap} is reduced to
\begin{align}\label{optimality_gap1}
F\left( {{{\bf{w}}_T}} \right) - {F^*} \le {\left( {1 - \frac{\delta }{L}} \right)^T}\left( {F\left( {{{\bf{w}}_0}} \right) - {F^*}} \right),
\end{align}
which becomes zero as $T \to \infty $. This confirms the usefulness of deploying IRSs to achieve a lossless FL model.

\vspace{-10pt}
\section{NOMA Versus OMA}
In this section, we provide a theoretical framework to compare the per-round latency achieved by the I-NOMA and I-OMA. For ease of illustration, we denote the optimal value of problem \eqref{C1}, \eqref{C2}, and \eqref{C3} as $\tau _{{\rm{TDMA}}}^*$, $\tau _{{\rm{FDMA}}}^*$, and $\tau _{{\rm{NOMA}}}^*$, respectively. First, the per-round latency achieved by the I-TDMA and I-FDMA is compared, whose relationship is illustrated in the following theorem.
\begin{thm}
It follows that $\tau _{{\rm{TDMA}}}^* \le \tau _{{\rm{FDMA}}}^*$, where the equality holds if and only if ${\bf{\Theta }}_1^* =  \ldots  = {\bf{\Theta }}_K^*$ with ${\bf{\Theta }}_k^*$ denoting the optimal IRS phase-shifts for device $k$ defined in \eqref{optimal_phaseshift}.
\end{thm}
\begin {proof}
For problem \eqref{C2} (i.e., the original optimization problem for I-FDMA), it is obvious that $\sum\nolimits_{k = 1}^K {{b_k} = 1}$ holds in constraint \eqref{C2-d} since otherwise we can always increase the value of $\sum\nolimits_{k = 1}^K {{b_k}}$ to reduce the objective value until $\sum\nolimits_{k = 1}^K {{b_k} = 1}$. Then, we introduce the slack variables as $E_{{\rm{F,}}k}^{\rm{c}} = {p_k}\tau _{\rm{F}}^c$, $\tau _{{\rm{F,}}k}^c = {b_k}\tau _{\rm{F}}^c$, $\forall k$. The objective function in \eqref{C2-a} is reduced to $\sum\nolimits_{k = 1}^K {\tau _{{\rm{F,}}k}^c + } \tau _t^{{\rm{loc}}}$ since $\sum\nolimits_{k = 1}^K {{b_k} = 1}$ holds. Then, problem \eqref{C2} can be equivalently transformed to
\begin{subequations}\label{C200}
\begin{align}
\label{C200-a}\mathop {\min }\limits_{{\bf{\Theta }},{\bf{a}},\left\{ {\tau _{{\rm{F,}}k}^c} \right\},\left\{ {E_{{\rm{F,}}k}^{\rm{c}}} \right\},{\bf{f}}} \;&\sum\nolimits_{k = 1}^K {\tau _{{\rm{F,}}k}^c + } \tau _t^{{\rm{loc}}}\\
\label{C200-b}{\rm{s.t.}}\;\;\;\;\;\;\;\;\;\;\;&{E_{{\rm{F,}}k}^{\rm{c}}} + \xi {C_k}{D_k}f_k^2 \le E_k^{\max },~\forall k,\\
\label{C200-c}&B\tau _{{\rm{F,}}k}^c{\log _2}\left( {1 + \frac{{E_{{\rm{F,}}k}^{\rm{c}}{{\left| {{h_k}\left( {\bf{\Theta }} \right)} \right|}^2}}}{{B\tau _{{\rm{F,}}k}^c{\sigma ^2}}}} \right) \nonumber\\
&\ge sa_k^t, \forall k,\\
\label{C200-f}&\eqref{C1-d},\eqref{C1-e},\eqref{C1-f},\eqref{C1-g}, \eqref{C2-e}.
\end{align}
\end{subequations}
It is evident that all feasible solutions of problem \eqref{C200} is also feasible to problem \eqref{C1}. Problem \eqref{C200} is equivalent to problem \eqref{C1} if and only if ${\bf{\Theta }}_1^* =  \ldots  = {\bf{\Theta }}_K^*$. Hence, $\tau _{{\rm{TDMA}}}^* \le \tau _{{\rm{FDMA}}}^*$ holds, which thus completes the proof.
\end{proof}

Theorem 2 indicates that I-TDMA is more preferable for the model uploading compared to I-FDMA since the IRS phase-shifts is time-selective rather than frequency-selective. Note that the time-selective channels created by the dynamic IRS beamforming is capable of further reducing the communication latency under the I-TDMA scheme. It is evident that I-NOMA always outperforms I-FDMA under the arbitrarily given IRS phase-shifts since NOMA is capacity achieving for uplink transmission. Next, we compare the per-round latency of I-TDMA and I-NOMA under the different setups. To highlight the impact of IRS, we focus on a homogeneous setup, where $E_1^{\max } =  \ldots  = E_K^{\max } = {E^{\max }}$, ${C_1} =  \ldots  = {C_K}$, ${D_1} =  \ldots  = {D_K}$. Besides, the direct AP-device link is assumed to be blocked, i.e., ${h_{d,k}} = 0$. Then, sufficient conditions for ensuring that I-TDMA outperforms I-NOMA and those of its opposite are unveiled in the following theorem.
\begin{thm}
It follows that $\tau _{{\rm{TDMA}}}^* \le \tau _{{\rm{NOMA}}}^*$ provided that
\begin{align}\label{TDMA_NOMA_condition}
\left| {{{\left[ {{{\bf{h}}_{r,1}}} \right]}_n}} \right| =  \ldots  = \left| {{{\left[ {{{\bf{h}}_{r,K}}} \right]}_n}} \right|,\forall n.
\end{align}
Besides, $\tau _{{\rm{TDMA}}}^* \ge \tau _{{\rm{NOMA}}}^*$ holds under the condition that
\begin{align}\label{NOMA_TDMA_condition}
\arg \left( {{\mathop{\rm diag}\nolimits} \left( {{{\bf{g}}^H}} \right){{\bf{h}}_{r,1}}} \right) =  \ldots  = \arg \left( {{\mathop{\rm diag}\nolimits} \left( {{{\bf{g}}^H}} \right){{\bf{h}}_{r,K}}} \right).
\end{align}
\end{thm}
\begin{proof}
Under the arbitrarily feasible scheduling device and computational frequency, i.e., $\left\{ {{f_k},{\cal K}_a^t} \right\}$, the remaining energy used for local model uploading for each scheduled device is $E_k^{\rm{c}} = {E^{\max }} - \xi {C_k}{D_k}f_k^2$. As shown in Proposition 1, we have ${f_1} =  \ldots  = {f_K}$ under the homogeneous setup, which further leads to $E_1^{\rm{c}} =  \ldots  = E_K^{\rm{c}}$. Since the computational latency is identical for I-TDMA and I-NOMA under the given $\left\{ {{f_k},{\cal K}_a^t} \right\}$, comparing $\tau _{{\rm{TDMA}}}^*$ and $\tau _{{\rm{NOMA}}}^*$ is equivalent to compare their uploading latency. In the following, we focus on the comparison of uploading latency of the I-TDMA and I-NOMA, denoted by $\tau _{{\rm{T}}}^{\rm{c}}$ and $\tau _{{\rm{N}}}^{\rm{c}}$, under the condition of \eqref{TDMA_NOMA_condition} and \eqref{NOMA_TDMA_condition}, respectively.

We first focus on the case of \eqref{TDMA_NOMA_condition}. Under the optimal IRS phase-shifts $\left\{ {{\bf{\Theta }}_k^*} \right\}$ of the I-TDMA, the equivalent channel power gain of each device is given by
\begin{align}\label{TDMA_channel_gain}
\gamma _k^* \!=\! {\left| {{{\bf{g}}^H}{\bf{\Theta }}_k^*{{\bf{h}}_{r,k}}} \right|^2} \!=\! {\sum\nolimits_{n = 1}^N {\left( {\left| {{{\left[ {\bf{g}} \right]}_n}} \right|\left| {{{\left[ {{{\bf{h}}_{r,k}}} \right]}_n}} \right|} \right)} ^2}.
\end{align}
Based on the condition \eqref{TDMA_NOMA_condition}, it follows that $\gamma _1^* =  \ldots  = \gamma _K^*$. The uploading latency of device $k$, denoted by $\tau _{{\rm{T,}}k}^{\rm{c}}$, satisfies the equation $B\tau _{{\rm{T,}}k}^{\rm{c}}{\log _2}\left( {1 + \frac{{E_k^{\rm{c}}\gamma _k^*}}{{B\tau _{{\rm{T,}}k}^{\rm{c}}{\sigma ^2}}}} \right) = s,\forall k \in {\cal K}_a^t$. Since $E_1^{\rm{c}}\gamma _1^* =  \ldots  = E_K^{\rm{c}}\gamma _K^*$, we obtain $\tau _{{\rm{T,}}k}^{\rm{c}} = \tau _{{\rm{T,}}l}^{\rm{c}},\forall k,l \in {\cal K}_a^t$, which leads to
\begin{align}\label{SNR_condition}
\frac{{E_k^{\rm{c}}\gamma _k^*}}{{B\tau _{{\rm{T,}}k}^{\rm{c}}{\sigma ^2}}} = \frac{{E_l^{\rm{c}}\gamma _l^*}}{{B\tau _{{\rm{T,}}l}^{\rm{c}}{\sigma ^2}}},\forall k,l \in {\cal K}_a^t.
\end{align}
It can be derived from \eqref{SNR_condition} that
\begin{align}\label{SNR_condition2}
&\sum\nolimits_{k \in {\cal K}_a^t} {B\tau _{{\rm{T,}}k}^{\rm{c}}{{\log }_2}\left( {1 + \frac{{E_k^{\rm{c}}\gamma _k^*}}{{B\tau _{{\rm{T,}}k}^{\rm{c}}{\sigma ^2}}}} \right)}\nonumber\\
&\mathop  = \limits^{\left( a \right)} B\tau _{\rm{T}}^{\rm{c}}{\log _2}\left( {1 + \frac{{\sum\nolimits_{k \in {\cal K}_a^t} {E_k^{\rm{c}}\gamma _k^*} }}{{B\tau _{\rm{T}}^{\rm{c}}{\sigma ^2}}}} \right) = \left| {{\cal K}_a^t} \right|s,
\end{align}
where $\tau _{\rm{T}}^{\rm{c}} = \sum\nolimits_{\forall k \in {\cal K}_a^t} {\tau _{{\rm{T,}}k}^{\rm{c}}}$ and (a) holds due to \eqref{SNR_condition}. Let ${{{\bf{\Theta }}^*}}$ denote the optimal IRS phase-shifts for the I-NOMA and ${{\tilde \gamma }_k} = {\left| {{{\bf{g}}^H}{{\bf{\Theta }}^*}{{\bf{h}}_{r,k}}} \right|^2}$. For the I-NOMA, we have $\tau _{\rm{N}}^{\rm{c}} \ge \tilde \tau _{\rm{N}}^{\rm{c}}$, where $\tilde \tau _{\rm{N}}^{\rm{c}}$ satisfies the equation
\begin{align}\label{equation_NOMA}
B\tilde \tau _{\rm{N}}^{\rm{c}}{\log _2}\left( {1 + \frac{{\sum\nolimits_{k \in {\cal K}_a^t} {E_k^{\rm{c}}{{\tilde \gamma }_k}} }}{{B\tilde \tau _{\rm{N}}^{\rm{c}}{\sigma ^2}}}} \right) = \left| {{\cal K}_a^t} \right|s.
\end{align}
Since ${{\tilde \gamma }_k} \le \gamma _k^*$ holds, we obtain $\tilde \tau _{\rm{N}}^{\rm{c}} \ge \tau _{\rm{T}}^{\rm{c}}$ by comparing \eqref{SNR_condition2} and \eqref{equation_NOMA}. Hence, $\tau _{\rm{T}}^{\rm{c}} \le \tilde \tau _{\rm{N}}^{\rm{c}} \le \tau _{\rm{N}}^{\rm{c}}$ naturally holds, which thus indicates that $\tau _{{\rm{TDMA}}}^* \le \tau _{{\rm{NOMA}}}^*$ under the condition of \eqref{TDMA_NOMA_condition}.

Next, we focus on the case of \eqref{NOMA_TDMA_condition}. When \eqref{NOMA_TDMA_condition} is satisfied, it is obvious that ${\bf{\Theta }}_1^* =  \ldots  = {\bf{\Theta }}_K^* = {{\bf{\Theta }}^*}$, which leads to $\gamma _k^* = {{\tilde \gamma }_k},\forall k$. The uploading latency of I-NOMA satisfies the equation $B\tau _N^{\rm{c}}{\log _2}\left( {1 + \frac{{\sum\nolimits_{k \in \tilde {\cal K}_a^t} {E_k^{\rm{c}}{{\tilde \gamma }_k}} }}{{B\tau _N^{\rm{c}}{\sigma ^2}}}} \right) = \left| {\tilde {\cal K}_a^t} \right|s$, where ${\tilde {\cal K}_a^t}$ is one subset of ${\cal K}_a^t$. Let $\tilde \tau _{\rm{T}}^{\rm{c}} = \sum\nolimits_{k \in \tilde {\cal K}_a^t} {\tau _{{\rm{T}},k}^{\rm{c}}}$ and then we have
\begin{align}\label{relaton_equality}
\left| {\tilde {\cal K}_a^t} \right|s &= \sum\nolimits_{k \in \tilde {\cal K}_a^t} {B\tau _{{\rm{T}},k}^{\rm{c}}{{\log }_2}\left( {1 + \frac{{E_k^{\rm{c}}\gamma _k^*}}{{B\tau _{{\rm{T}},k}^{\rm{c}}{\sigma ^2}}}} \right)}\nonumber\\
&\mathop  \le \limits^{\left( a \right)} B\sum\nolimits_{k \in \tilde {\cal K}_a^t} {\tau _{{\rm{T}},k}^{\rm{c}}{{\log }_2}\left( {1 + \frac{{\sum\nolimits_{k \in \tilde {\cal K}_a^t} {E_k^{\rm{c}}{{\tilde \gamma }_k}} }}{{B{\sigma ^2}\sum\nolimits_{k \in \tilde {\cal K}_a^t} {\tau _{{\rm{T}},k}^{\rm{c}}} }}} \right)}\nonumber\\
& = B\tilde \tau _{\rm{T}}^{\rm{c}}{\log _2}\left( {1 + \frac{{\sum\nolimits_{k \in \tilde {\cal K}_a^t} {E_k^{\rm{c}}{{\tilde \gamma }_k}} }}{{B\tilde \tau _{\rm{T}}^{\rm{c}}{\sigma ^2}}}} \right),
\end{align}
where (a) holds because $x{\log _2}\left( {1 + y/x} \right)$ is joint a concave function with respect to $x$ and $y$. From \eqref{relaton_equality}, we obtain $\tilde \tau _{\rm{T}}^{\rm{c}} \ge \tau _N^{\rm{c}}$. Since $\tilde {\cal K}_a^t \subseteq {\cal K}_a^t$ and $\tau _{\rm{T}}^{\rm{c}} = \sum\nolimits_{k \in {\cal K}_a^t} {\tau _{{\rm{T}},k}^{\rm{c}}}$, $\tau _N^{\rm{c}} \le \tilde \tau _{\rm{T}}^{\rm{c}} \le \tau _{\rm{T}}^{\rm{c}}$ naturally holds, which thus completes the proof.
\end{proof}

Theorem 3 implies that I-TDMA outperforms I-NOMA provided that condition \eqref{TDMA_NOMA_condition} is satisfied. Note that condition \eqref{TDMA_NOMA_condition} refers to the \emph{power homogeneous setting}, which holds in practice when all devices are located on a half-circle centered at the IRS under the LoS channel setup. In contrast, \eqref{NOMA_TDMA_condition} in Theorem 3 serves as a sufficient condition for that I-NOMA outperforms I-TDMA. In particular, condition \eqref{NOMA_TDMA_condition} refers to the \emph{phase homogeneous setting}, which holds in practice when all devices are located in a line passing through the IRS. By considering the nature of IRS, the results in Theorem 3 are fundamentally different from the conclusions in previous works, e.g., \cite{mo2021energy, bouzinis2022wireless}, which demonstrated that NOMA always outperforms OMA in terms of latency of wireless FL system without IRS.
\vspace{-10pt}
\section{Simulation Results}
In this section, simulation results are provided to validate our analytical results (Theorems 2 and 3) and demonstrate the usefulness of using IRS to improve the performance of wireless FL. We consider a two dimensional coordinate setup measured in meter (m), where the AP and IRS are located at $\left( {0,0} \right)$ m and $\left( {100,5} \right)$ m, respectively. The devices are distributed in the vicinity of the IRS and the detailed distributions will be specified in the following discussions. For the involved channels, the pathloss exponents for both the AP-IRS and IRS-device links are set to $2$, whereas the pathloss exponent of the direct AP-device link is set to $3.4$. For each individual link, the pathloss at the reference distance of 1 m is set to 30 dB. Unless otherwise stated, other parameters are set as follows: $B = 10$ MHz, $B{\sigma ^2} =  - 80$ dBm, $\zeta  = {10^{ - 27}}$, $s = 1$ Mbit, and ${C_k} = 10$ cycles.
\vspace{-10pt}
\subsection{Communication Latency Comparison}
In this subsection, we focus on the communication latency to provide an in-depth discussion regarding the performance comparison among I-TDMA, I-FDMA, and I-NOMA under various types of channel setups. To show it clearly, we consider 10 devices, i.e., $K = 10$, to upload their models to the AP and aim to minimize the associated communication latency of different multiple access schemes.

\begin{figure}[!t]
 \centerline{\includegraphics[width=2.4in]{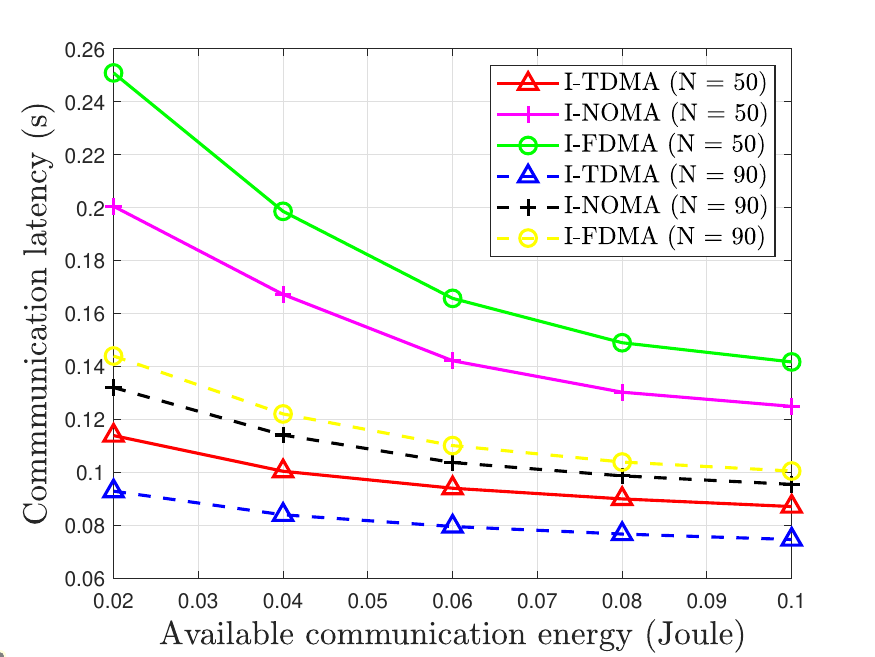}}
 \caption{Communication latency versus communication energy under the power homogeneous setting.}
 \label{C_latency1}
 \vspace{-8pt}
\end{figure}

\subsubsection{Power Homogeneous Setting}
We first study the performance comparison between the OMA and NOMA under the power homogeneous setting. In a power homogenous setting, we consider a pure-LoS channel setup and the AP-device direct links are blocked. Besides, all devices are uniformly distributed at a half-circle centered at the IRS with a radius of $10$ m, which leads to condition \eqref{TDMA_NOMA_condition}. In Fig. \ref{C_latency1}, we plot the achieved communication latency by I-TDMA, I-NOMA, and I-FDMA versus the available communication energy under this power homogeneous setting. Apparently, it is observed that the communication latency of all schemes can be significantly reduced by increasing the available communication energy or the number of IRS elements. Moreover, one can observe that the latency achieved by I-TDMA is lower than those achieved by I-NOMA and I-FDMA. This is due to the dynamic adjustment of IRS phase-shifts to cater for the channels of different devices by exploiting time-selectivity of IRS. Besides, I-FDMA always achieves a higher communication latency than the other two counterparts. The results in Fig. \ref{C_latency1} validates the analysis in Theorems 2 and 3.

\begin{figure}[!t]
 \centerline{\includegraphics[width=2.4in]{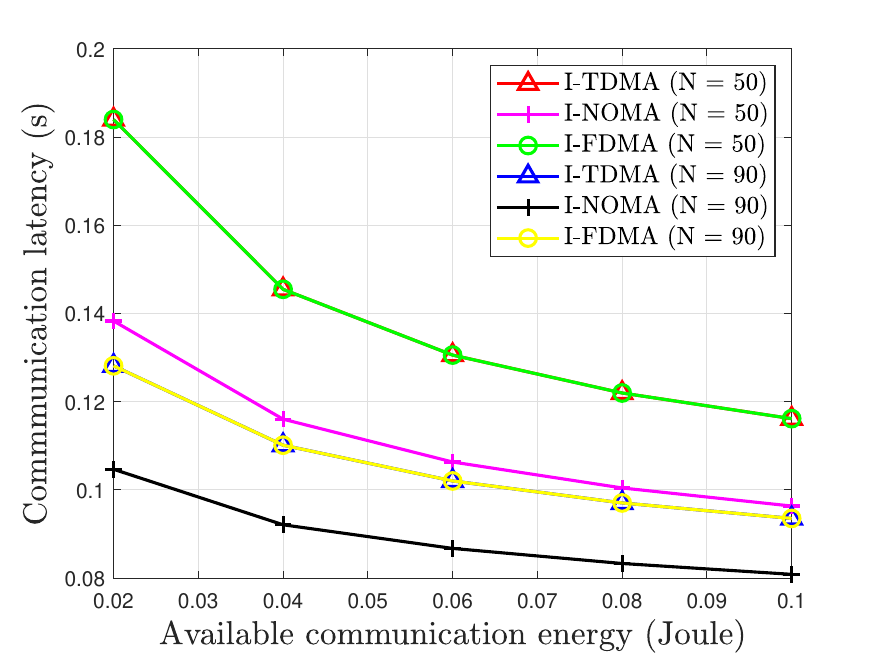}}
 \caption{Communication latency versus communication energy under the phase homogeneous setting.}
 \label{C_latency2}
 \vspace{-8pt}
\end{figure}

\subsubsection{Phase Homogeneous Setting}
Then, we further investigate the communication latency of different multiple access schemes under the phase homogeneous setting. In a phase homogenous setting, a pure LoS channel setup is considered and AP-device links are blocked. All devices are uniformly distributed at a line segment characterized by $x = 100, - 30 \le y \le 0$, which leads to condition \eqref{NOMA_TDMA_condition}. Under this setting, we plot the communication latency versus the available communication energy in Fig. \ref{C_latency2}. First, it is observed that the I-TDMA and I-FDMA achieves the identical communication latency, which is in accordance with Theorem 2. In a phase homogeneous setting, the cascade channels of all devices share the same phase and thus the optimal IRS phase-shift matrixes to maximize each device's channel power gain is identical. Under the identical channel setup by setting a set of common IRS phase-shifts, I-TDMA achieves the same communication latency as that of I-FDMA. Moveover, we observe that I-NOMA achieves a lower communication latency than the other two counterparts, which agrees the analysis in Theorem 3. The result is expected since NOMA is a capacity achieving scheme for the uplink transmission under the static channel setups.

\subsubsection{General Setting}
\begin{figure}[!t]
 \centerline{\includegraphics[width=2.4in]{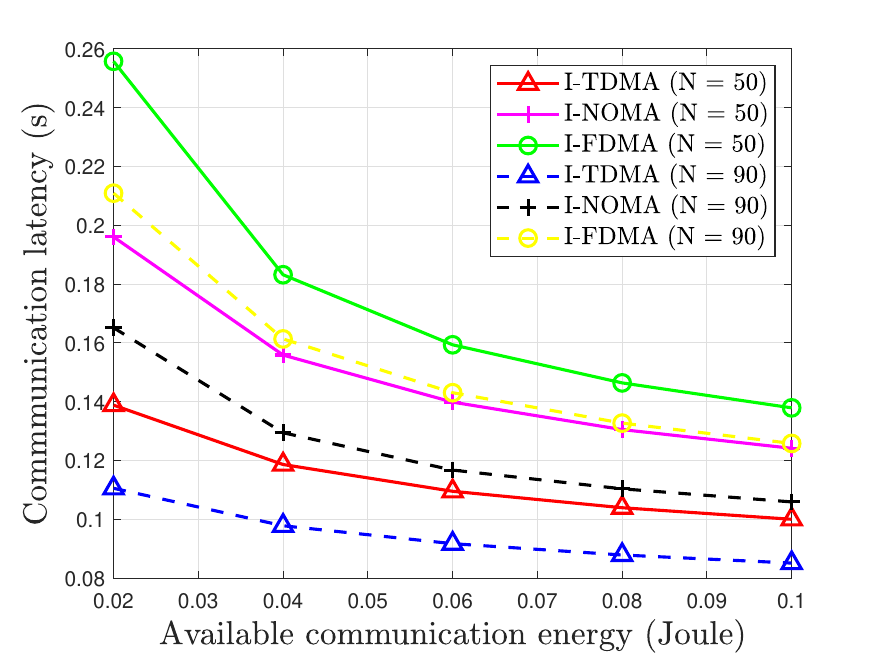}}
 \caption{Communication latency versus communication energy under the general setting.}
 \label{C_latency3}
 \vspace{-8pt}
\end{figure}

Finally, we examine the communication latency under a general channel setting. In this setup, all devices are uniformly and randomly distributed within a radius of $20$ m centered at the IRS. Furthermore, Rician fading with Rician factor of $3$ dB is used as the small-scale fading for all channels. The communication latency versus the available communication latency is depicted in Fig. \ref{C_latency3} under this setup. It is observed from Fig. \ref{C_latency3} that I-TDMA outperforms I-NOMA and I-FDMA in terms of the communication latency. The reason is that the both the I-NOMA and I-FDMA schemes use a single set of  IRS phase-shifts for assisting model uploading of multiple devices, which is less flexible for channel reconfiguration. In contrast, multiple IRS reflection patterns employed in the I-TDMA scheme provides more degree of freedoms for further enhancing the passive beamforming gain for each individual device. The higher passive beamforming gain attained for the I-TDMA becomes a more dominant factor to reduce the communication latency. The results confirm the superiority of I-TDMA due to the dynamic IRS beamforming.

\vspace{-10pt}
\subsection{IRS Enhanced Wireless FL}
In this subsection, we evaluate the usefulness of IRS for enhancing the performance of wireless FL in terms of the per-round latency and test accuracy. To show the performance of the proposed design for handling specific FL tasks, the image classifier model on the widely used MNIST datasets is trained. In particular, we consider $K=20$ devices are randomly distributed within a radius of 20 m centered at the IRS and the channel model is set as that in the general setting of the previous subsection. For the imbalanced data numbers, we divide $20$ devices into two parts, and they have 1000 and 2000 samples for each device, respectively. Based on the superiority of I-TDMA over I-NOMA and I-FDMA demonstrated in the previous subsection, we adopt the I-TDMA to represent our proposed design.

The proposed design is compared with the following benchmark schemes: 1) Full scheduling: All devices are scheduled in each training round, i.e., $a_k^t = 1,\forall k$, and the remaining variables are optimized; 2) Random IRS phase-shifts: Remaining varialbes are optimized under the random IRS phase-shifts; 3) SNR-based scheduling: The device scheduling is performed based on SNR to satisfy constraint \eqref{weight_deivice} and the remaining variables are optimized by using the proposed algorithm. 4) Without IRS: Resource allocation and device scheduling are optimized without the IRS.

\begin{figure}[!t]
 \centerline{\includegraphics[width=2.4in]{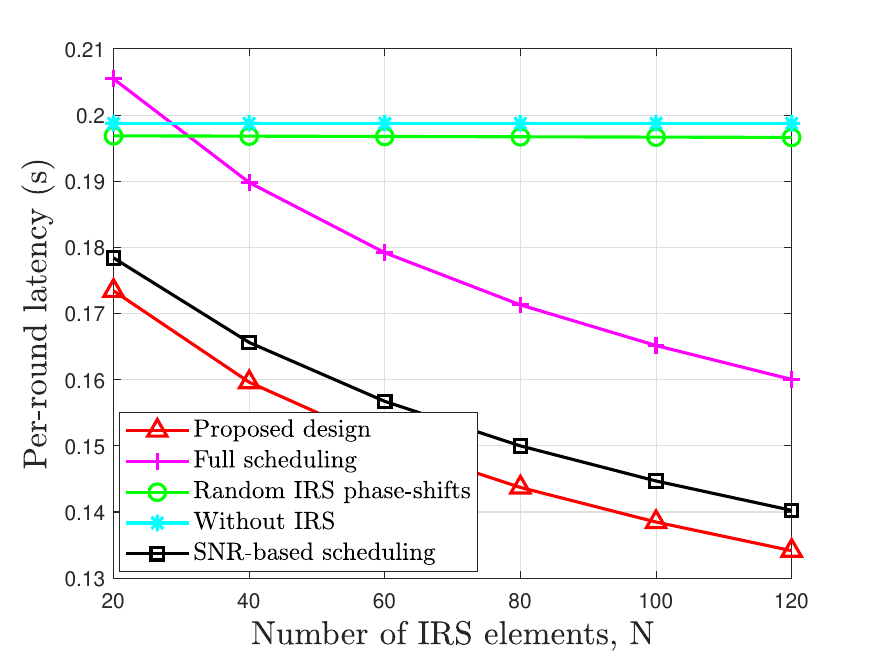}}
 \caption{Per-round latency versus $N$ with $E = 0.1$ Joule.}
 \label{round_latency1}
 \vspace{-8pt}
\end{figure}

\subsubsection{IRS for Reducing Training Latency}
In Fig. \ref{round_latency1}, we plot the per-round latency versus $N$ with the training loss $\nu  = 0.15$, which is defined by $\nu  = \sum\nolimits_{k = 1}^K {\left( {1 - a_k^t} \right){D_k}} /D$.  First, it is observed that the per-round latency achieved by the schemes with IRS phase-shifts optimization monotonically decreases with $N$. The reasons are two-folds. On the one hand, deploying more IRS elements helps achieve a higher passive beamforming gain, which improves the uplink spectral efficiency. On the other hand, a higher passive beamforming gain reduces the demand for communication energy, which leaves more energy allocated to local model training to reduce the computational latency. Besides, the per-round latency with random IRS phase-shifts is not sensitive to increasing $N$ and its gain over the scheme without IRS is marginal, whereas the scheme with full scheduling even performs worse than the scheme without IRS for small $N$, but significantly outperforms the scheme with random phase-shifts for large $N$. The result is expected since a higher passive beamforming gain helps compensate the additional latency incurred by scheduling more devices. Moreover, our proposed design outperforms the scheme with SNR-based scheduling since both the issues of wireless channel and data volume are considered.

\begin{figure}[!t]
 \centerline{\includegraphics[width=2.4in]{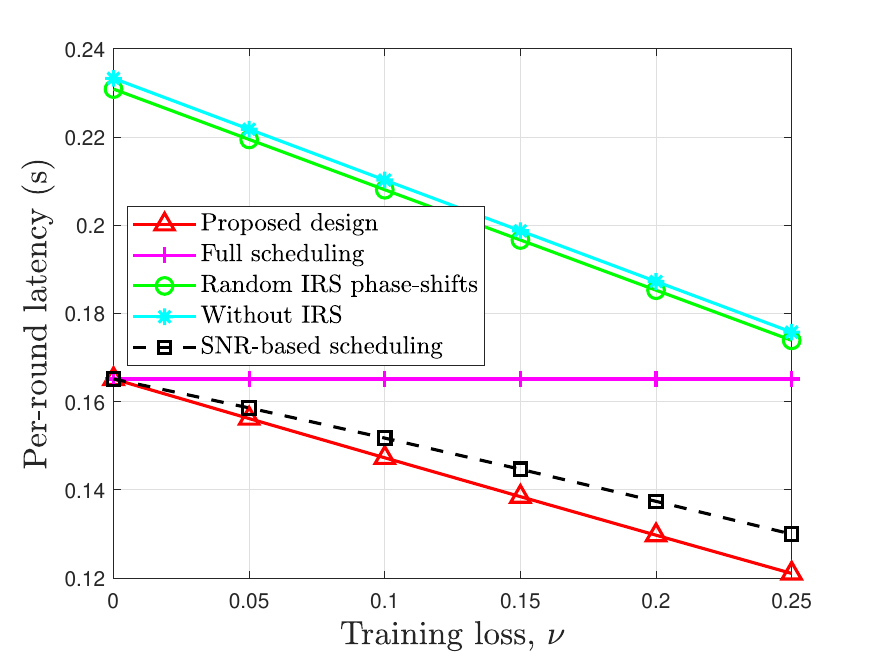}}
 \caption{Per-round latency versus the training loss $\nu $ with $E = 0.1$ Joule.}
 \label{round_latency2}
 \vspace{-8pt}
\end{figure}

Fig. \ref{round_latency2} shows the impact of the target training loss $\nu$ on the per-round latency when $N=100$. As expected, the performance of full scheduling  is invariant to $\nu  $ since the number scheduled device is fixed under the full scheduling scheme. For the other schemes considering device scheduling, the per-round latency monotonically decreases with $\nu  $, which implies a fundamental tradeoff between the learning latency and learning accuracy. Besides, it is observed that the latency achieved by the full scheduling is still lower than that of the scheme with random phase-shifts even when  $\nu = 0.15$, which underscores the usefulness of using IRS to enhance the latency-accuracy tradeoff in wireless FL systems. Moreover, the proposed design significantly outperforms other benchmark schemes, which highlights the importance of the joint optimization of IRS phase-shifts and device scheduling.

\subsubsection{IRS Enhanced Training Accuracy}
Then, we discuss on using IRS to improve the training accuracy of wireless FL under the given per-round latency. The results are obtained by solving problem \eqref{C18} in Section IV-D. Under the target per-round latency $\bar \tau  = 0.15~{\rm{s}}$, we study the impact of the number of IRS elements on the average number of scheduling devices in Fig. \ref{learning_quality1}. By increasing the number of IRS elements, the system is capable of scheduling more devices due to the attained higher passive beamforming gain, which is beneficial for improving data exploitation. In particular, it is observed the energy budget of devices for achieving a full scheduling can be reduced from 0.4 Joule to 0.2 Joule by increasing $N$ from $N = 60$ to $N = 120$. The result implies that IRS is able to support low-energy devices to achieve a high-quality learning performance.

\begin{figure}[!t]
 \centerline{\includegraphics[width=2.4in]{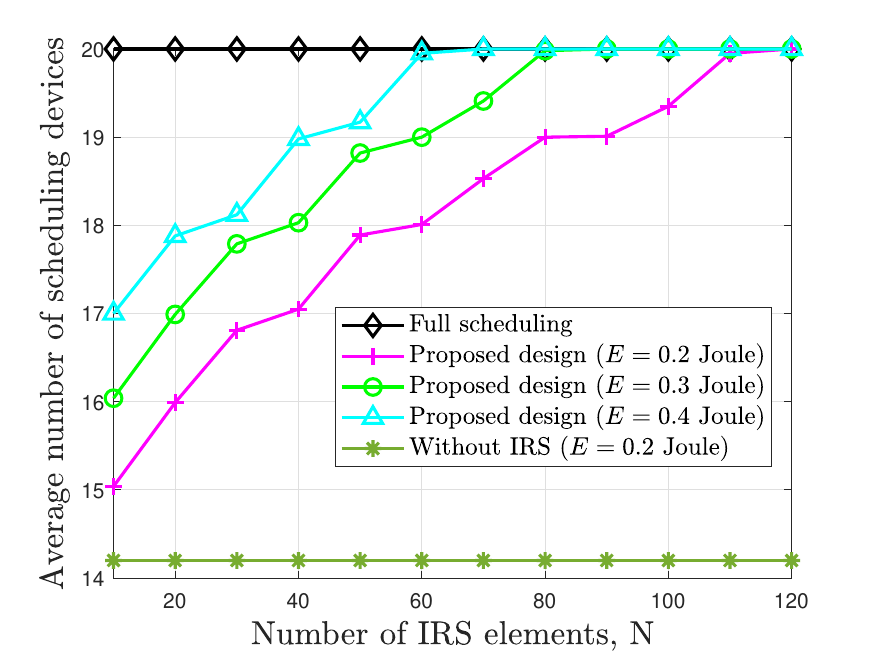}}
 \caption{The number of scheduling devices versus $N$.}
 \label{learning_quality1}
 \vspace{-8pt}
\end{figure}

\begin{figure}[!t]
 \centerline{\includegraphics[width=2.4in]{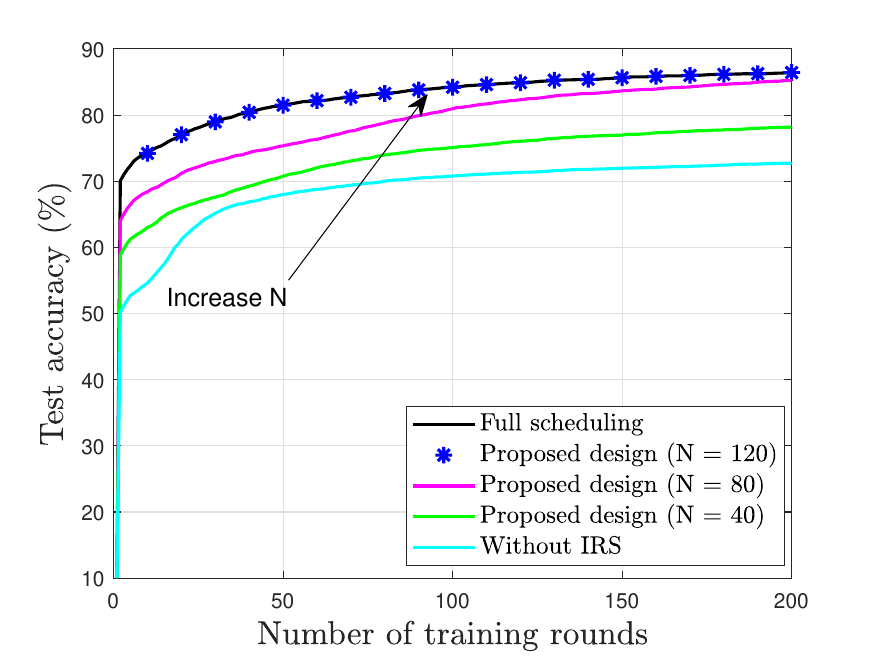}}
 \caption{Test accuracy versus the number of training rounds.}
 \label{learning_quality2}
  \vspace{-8pt}
\end{figure}

To illustrate the performance of the proposed joint design for handling FL tasks, a image classifier model is trained on the widely-adopted MNIST dataset. We plot the test accuracy versus the number of training rounds in Fig. \ref{learning_quality2}. Note that the scheme of full scheduling represents an upper bound of test accuracy since the model aggregation error is not introduced. It is observed that the test accuracy of the proposed design approaches that achieved by the full scheduling scheme as the increase of $N$, which is expected since more devices are scheduled to contribute their data exploitation, thereby improving the test accuracy significantly. The result confirms the analysis in Proposition 4, which demonstrates the usefulness of deploying IRS to achieve a lossless FL model.
\vspace{-10pt}
\section{Conclusion}
This paper proposed three transmission protocols, namely I-TDMA, I-FDMA, and I-NOMA for the local model uploading in IRS aided wireless FL systems. Under the proposed three protocols, we jointly optimized the IRS phase-shifts and resource allocation to minimize the per-round latency. The proposed optimization framework was further extended to the training loss minimization problem under the given latency requirement to shed light on the fundamental tradeoff between the learning accuracy and learning latency. The required number of IRS elements to enable a full scheduling was derived. Then, we provided a theoretical framework to compare the performance of the proposed three protocols. Both the preferable channel settings for the I-NOMA and I-TDMA were unveiled, which is fundamentally different from that NOMA always outperforms OMA in wireless FL systems without IRS.


\bibliographystyle{IEEEtran}
\vspace{-10pt}
\bibliography{IEEEabrv,myref}


\end{document}